\documentclass[11pt, oneside]{article} 
\usepackage[margin=2cm]{geometry}                		
\usepackage[parfill]{parskip}    		
\usepackage{graphicx}			
\usepackage[utf8]{inputenc}
\usepackage[english]{babel} 
\usepackage{amsmath,amssymb,amsthm,bm} 
\usepackage[margin=2cm]{geometry}
\usepackage[color=yellow]{todonotes}
\usepackage{mathrsfs}
\usepackage{url}	
\usepackage{esint}
\usepackage[colorlinks]{hyperref}
\usepackage{enumerate}
\usepackage{outlines}[enumerate]
\usepackage{framed,comment,enumerate}
\usepackage{upgreek}
\usepackage{pgfplots}
\usepackage{float}
\usepackage{tikz,tikz-cd}
\usepackage{rotating}
\usepackage{graphicx}
\usepackage{caption}
\usepackage{multicol}
\usepackage{cancel}
\usepackage{subcaption}
\usepackage[title]{appendix}
\usepackage{tikz-cd}
\usepackage{pgfgantt}
\usepackage{rotating}
\extrafloats{100}
\numberwithin{equation}{section}

\newtheorem{proposition}{Proposition}[section]
\newtheorem{remark}{Remark}[section]

\theoremstyle{definition}

\DeclareFontFamily{U}{MnSymbolC}{}
\DeclareSymbolFont{MnSyC}{U}{MnSymbolC}{m}{n}
\DeclareFontShape{U}{MnSymbolC}{m}{n}{
    <-6>  MnSymbolC5
   <6-7>  MnSymbolC6
   <7-8>  MnSymbolC7
   <8-9>  MnSymbolC8
   <9-10> MnSymbolC9
  <10-12> MnSymbolC10
  <12->   MnSymbolC12}{}
\DeclareMathSymbol{\intprod}{\mathbin}{MnSyC}{'270}

\newcommand{\wt}[1]{\widetilde{#1}}
\newcommand{\wh}[1]{\widehat{#1}}
\newcommand{\mb}[1]{\mathbf{#1}}

\newcommand{\mc}[1]{\mathcal{#1}}
\newcommand{\bs}[1]{\boldsymbol{#1}}

\newcommand{\mcal}[1]{\mc{#1}}
\newcommand{\scp}[2]{\left<#1\,,\,#2\right>}

\newcommand{\ad}{\operatorname{ad}}
\newcommand{\dn}{\mathrm{d}}

\DeclareMathOperator{\diff}{d\!}

\def\p{{\partial}}

\def\rmd{{\rm d}}

\def\bA{{\mathbf{A}}}
\def\bB{{\mathbf{B}}}
\def\bJ{{\mathbf{J}}}

\def\bu{{\mathbf{u}}}
\def\bv{{\mathbf{v}}}
\def\bx{{\mathbf{x}}}

\def\p{\partial}

\pgfplotsset{compat=1.16}

\def\nti{\tilde{n}}

\def\bu{\mathbf{u}}
\def\bv{\mathbf{v}}
\def\bx{\mathbf{x}}

\begin{document}

\title{\textbf{Deterministic and Stochastic Geometric Mechanics for Hall MHD}}
\author{Darryl D. Holm, Ruiao Hu and Oliver D. Street \\ 
Mathematics, Imperial College London \\ \footnotesize
d.holm@ic.ac.uk, ruiao.hu15@imperial.ac.uk, o.street18@imperial.ac.uk 
\\  \small
\small
Lie group invariant variational principles; Magnetohydrodynamics; Hall effect
}
\date{}

\maketitle

\begin{abstract}
We derive new models of stochastic Hall magnetohydrodynamics (MHD)
by using a symmetry-reduced stochastic Euler-Poincar\'e variational principle. The new stochastic
Hall MHD theory has potential applications for uncertainty quantification and data 
assimilation in space plasma (space weather) and solar physics. The stochastic geometric
mechanics approach we take here produces coordinate-free results which may then be applied in a variety of 
spatial configurations. 
\end{abstract}

\tableofcontents

\section{Introduction} \label{sec-Intro}
\paragraph{Modelling uncertainty in Hall magnetodynamics (MHD) with Transport Noise}
This paper uses transport noise in a variational formulation to address the uncertainty produced by the instabilities and nonlinear complexity of compressible and ideal Hall MHD. 

Transport Noise is a new concept in stochastic fluid dynamics which has been applied to quantify uncertainty and assimilate data for fluid dynamics \cite{Holm2024}. The mathematical basis of Transport Noise is a stochastic version of Hamilton's variational principle \cite{Holm2015} which has been extended to include geometric rough paths \cite{CHLN2022a,CHLN2022b}. New results in mathematical analysis \cite{Crisan-etal-2023a,Crisan-etal-2023b} and applications to data assimilation \cite{Cotter-etal-2019,Cotter-etal-2020} have followed these mathematical advances, especially for modelling physically important effects in oceanography. For example, new results have been obtained via the variational approach for quantifying uncertainty in models of wave-current interaction focusing on the nonlinear effects of waves accelerating currents in Geophysical Fluid Dynamics (GFD) \cite{HHS2022a,HHS2022b}.

Here, we derive deterministic and stochastic Hall MHD systems by applying stochastic geometric mechanics (SGM). SGM relies on a recently discovered \cite{Holm2015} stochastic variational Hamilton's principle which extends the celebrated Kraichnan model of turbulence based on passive stochastic advection \cite{RHK1991} to further include the effects of Stochastic Advection by Lie Transport (SALT) in the fluid motion equation. This extension, in turn, modifies the deterministic Kelvin-Noether circulation theorem of ideal fluid dynamics \cite{HMR1998}. The sources of the recent development of \textit{Transport Noise} in mathematical statistics are summarised in \cite{Holm2024}. In particular, the SALT approach of \cite{Holm2015} has recently been extended to include geometric rough paths, in \cite{CHLN2022a,CHLN2022b}. In addition, a series of three papers using the SALT approach to derive new GFD models of wave-current interaction has appeared in \cite{Holm2019,HHS2022a,HHS2022b}.

\paragraph{Overview}  The Hall effect in a neutral plasma is a classical problem in plasma physics \cite{Brushlinsky1975}. The Hall effect produces an additional advective drift of the magnetic field lines 
induced by the electron fluid motion. The Hall drift velocity is proportional to the current density, which generates the magnetic force, that in turn generates fluid vorticity leading to MHD turbulence. In particular, it alters
the magnetic field configuration by creating small-scale structures that seed high-wavenumber instabilities and turbulence. 

Geometric mechanics models Hall MHD as a two-fluid ion-electron plasma, \cite{Holm1987,HK1987}. The nonlinear interactions of the slow-fast, large-small, resolved-unresolved decomposition of the mean and the fluctuating or turbulent components of the physical processes in MHD plasmas require a structure-preserving stochastic approach.    
This type of modelling is accomplished by separating the ion-electron plasma dynamics into two maps that whose composition governs the momentum densities the ions and the electrons.

\paragraph{Summary} The paper will use Lie-group invariant variational principles of geometric mechanics to introduce transport noise for the purpose of deriving and investigating the effects of stochastic advection in the interaction of the two-fluid representation of Hall MHD turbulence in neutral ion-electron plasmas.

\paragraph{Plan} This paper takes two approaches in geometric mechanics to derive deterministic and stochastic variational formulations of the PDE system for ideal Hall MHD.  
These two approaches reveal complementary features as we derive the equations for stochastic compressible Hall MHD flow and show that these stochastic equations preserve the Kelvin-Noether conservation laws for circulation and helicity. Section \ref{sec-HamPont} takes the Hamilton-Pontryagin approach, which emphasises the roles of the two flow maps underlying Hall MHD. Section \ref{sec-EulPoinc} takes the Euler-Poincar\'e approach which applies a more direct Lie  algebraic method to derive the equations. Section \ref{sec-SHMHD} contains the derivation of stochastic Hall MHD which introduces independent transport noises into both ion and electron fluids, taking  care to respect the algebraic condition of local charge neutrality. This is our main result. We derive this result using coordinate-free calculus of differential forms to ensure its applicability in any simply connected spatial configuration.  Section \ref{sec: SFLT} discusses an alternative to the SALT approach of introducing stochasticity. The alternative approach is known as Stochastic Forcing by Lie Transport (SFLT). As opposed to SALT, the SFLT approach has the feature that if preserves total integrated energy. Section \ref{sec-remarks} contains remarks, open problems and plans for future work, including further development of the SFLT approach of introducing stochasticity into fluid dynamics.

\paragraph{Kindred papers} Finally, we mention a few recent papers which seem aligned with either the approach or the goals of the present work. On one hand, the work closest to ours in terms of the geometric mechanics approach for two-fluid MHD is \cite{Besse2023}. However, the goals in \cite{Besse2023} concern stochastic Lagrangian formulations of \emph{deterministic} (not stochastic) dissipative PDEs. In contrast to the goals in \cite{Besse2023}, the present work seeks instead to develop variational principles that can help formulate \emph{stochastic} PDEs for parameterising hydromagnetic fluid flows. 

On the other hand, the goals and discussions of the onset of line-element stochasticity in hydromagnetic systems in \cite{Eyink2009} do seem somewhat aligned with the goals of the present work. Beyond discussions of Hall MHD, the paper \cite{Eyink2009} also discusses the two-fluid plasma model known as the Braginsky equations \cite{Brag1965,BragRoberts1995,GV2021}. Braginsky equations with stochastic transport would be well within the range of the stochastic variational principles considered in the present paper.
Another recent paper that seems conceptually related to the present work is \cite{Eyink-etal2013}. However, the relation of the present work to \cite{Eyink-etal2013} is a bit more tenuous for two reasons. In particular, the approach of \cite{Eyink-etal2013} is not variational. Conversely, the present work does not deal with either `anomalous diffusion' or `spontaneous stochasticity'. The latter terms refer to turbulence phenomena, rather than to the stochastic variational principles developed here.

\section{Deterministic Hamilton--Pontryagin principles for Hall MHD}\label{sec-HamPont}

\subsection{Hall MHD for ideal incompressible flow}
After assuming local charge neutrality, the Hall MHD equations for ideal inhomogeneous volume preserving flow 
are given in vector calculus form by \cite{Goedbloed2019}
\begin{align}
    \begin{split}
        &\p_t \rho\bu + \bu\cdot\nabla \rho\bu 
        + \nabla p + \mb{B} \times \operatorname{curl}\mb{B} = 0\,,
        \\
        &\p_t \rho + \mathrm{div}(\rho\bu) = 0\,,\quad\mathrm{div}\bu = 0
        \\
        &\p_t \mb{A} + \mb{B}\times \bv = 0
        \,,\quad \mb{B}=\mathrm{curl}\bA
        \,,\\&
        \bv:= \bu - \frac{R}{a\rho}\operatorname{curl}\mb{B}
        = \bu + \bv_H \,, \quad \text{with} \quad \bv_H:= - \frac{R}{a\rho}\operatorname{curl}\mb{B} = -\frac{R}{a\rho} \bJ
        \,.
    \end{split}
    \label{HMHD-vector-eqns}
\end{align}

Here, $\bu$ is the divergence-free fluid flow velocity, $\rho$ is the local fluid density, $p$ is the fluid  pressure, $\bJ=\mathrm{curl}\bB$ is the magnetic current density, $\bA$ is the magnetic vector potential, with magnetic field $\bB = \mathrm{curl} \bA$. The Hall velocity $\bv_H$ is proportional to the Hall parameter, $R$, determined as the ratio of the Larmor radii of the electrons and ions. The charge-to-mass ratio of electrons and ions is given by $a$ (a negative number). The MHD equations are recovered for $R=0$.  

Boundary conditions for \eqref{HMHD-vector-eqns} require vanishing of normal coordinates of $\bu$, $\nabla p$, $\bB$ and $\bJ$ on the boundary. In the radiation gauge one sets $\p_t\bA\times\mathbf{\hat{n}}|_{\p D}=0$. Hence, if $\bA\times\mathbf{\hat{n}}|_{\p {\cal D}}=0$ holds initially, then it will remain so. Here, $\mathbf{\hat{n}}$ is the normal unit vector on the boundary $\p {\cal D}$ of the domain ${\cal D}$.  See  \cite{Holm1987} for further discussion.

\paragraph{Coordinate-free geometric notation for Hall MHD}
The coordinate-free geometric form of the equation set in \eqref{HMHD-vector-eqns} is 
\begin{align} 
\begin{split}
\left( \p_t + \mathcal{L}_u \right) u^\flat
&= - \frac{1}{\rho}\mb{d} p + \mb{d}\left( \frac{1}{2} u\intprod u^\flat\right) + \frac{1}{\rho}\left((v-u)\intprod \frac{a}{R} \mb{d}A \right)
\,,\\
\left( \p_t + \mathcal{L}_u \right)\rho &= 0
\,,\\
\left( \p_t + \mathcal{L}_{v} \right) A &= \mb{d} (v \intprod A)
\,, \quad\hbox{with}\quad B=\mb{d} A
\,,\\
\mathcal{L}_u (d^3x) &= \mb{d}(u \intprod d^3x) = (\mathrm{div}\bu)d^3x = 0 
\,,\\
\hbox{and}\quad \mathrm{div}\bu &= 0\,, \quad \mathrm{div}\bB = 0 
\,.
\end{split}
\label{HMHD-GMeqns}
\end{align} 

\begin{remark}[Coordinate-free geometric notation]  In the equation system \eqref{HMHD-GMeqns}, the exterior differential operator is denoted as $\mb{d}$ and the musical superscript notation $(\,\cdot\,)^\flat : \mathfrak{X}(\mcal{D}) \rightarrow \Lambda^1(\mcal{D})$ identifies the transformation of vector fields to 1-forms. The inverse transformation from 1-forms to vector fields is written in musical notation as $(\,\cdot\,)^\sharp : \Lambda^1(\mcal{D}) \rightarrow \mathfrak{X}(\mcal{D})$. The notation $v \intprod A$ denotes the insertion of a vector field $v$ into the 1-form $A$. Likewise, the notation $u \intprod u^\flat$ denotes the insertion of a vector field $u$ into the 1-form $u^\flat$. In Euclidean coordinates, this particular insertion would be the dot-product. Of course, also the differential $\mb{d} B= (\mathrm{div}\bB)d^3x=\mb{d}^2 A = 0$ because $B=\mb{d} A$ is an exact 2-form. 
In addition, the operation $\mathcal{L}_u$ is defined by
\begin{align}
    \mathcal{L}_u a:= u \intprod \mb{d}a + \mb{d}(u \intprod a)\,,\quad \forall a \in \Lambda^k(\mcal{D})\,,
\label{LieDer-def}
\end{align} 
which is the Cartan's form of the Lie derivative $\mathcal{L}_u$ along the vector field $u$ on k-forms.
\end{remark}
\begin{remark}[Division by a density]\label{remark:dividing_by_density}
    Suppose that we have a $1$-form density, $\mu = \bs{\mu}\cdot d\bx\otimes dV$. One may formally define division of $1$-form density, $\mu$ by a density $\nti:= *\nti\, dV$ (where $\star\nti\in \Lambda^0(\mcal{D})$ is the Hodge dual of $\nti\in\Lambda^d(\mcal{D})$ in $d$ dimensions) to mean the following,
    \begin{equation*}
    \mu = \bs{\mu}\cdot d\bx\otimes dV = \frac{\bs{\mu}}{\star\nti}\cdot d\bx\otimes \nti 
    \quad\hbox{so that}\quad\frac{\mu}{\nti} := \frac{\bs{\mu}}{\star\nti}\cdot d\bx \,.
    \end{equation*}
\end{remark}
\begin{remark}[Conservation laws]
In the \emph{homogeneous} case when $\rho=const$ equations \eqref{HMHD-GMeqns} imply conservation laws for two circulations and two helicities:
\begin{align} 
\begin{split}
\oint_{c(u)} u^\flat + \frac{a}{R}A \quad\hbox{(Total circulation)}
&\quad\hbox{and}\quad \oint_{c(v)} A  \quad\hbox{(Magnetic circulation)}
\,,\\ 
\int_{\cal D} C\wedge \mb{d} C \quad\hbox{(Total helicity)}
&\quad\hbox{and}\quad  \int_{\cal D} A\wedge \mb{d} A \quad\hbox{(Magnetic helicity)} \,,
\end{split}
\label{HMHD-GMeqns4}
\end{align}
in which $\wedge$ denotes the wedge product of differential forms and the 
variable $C := \mathbf{C}\cdot \mb{d} \bx$ is the \emph{total circulation 1-form} defined by
\begin{align}
C:= u^\flat + \frac{a}{R}A \,.\label{tot-circ-vector}
\end{align} 

In the \emph{inhomogeneous} case when $\rho$ is advected by the fluid velocity vector field $u$ with Euclidean vector components, $\bu$, the equations in \eqref{HMHD-GMeqns} imply conservation laws for only magnetic circulation and magnetic helicity in \eqref{HMHD-GMeqns4}.

In the \emph{inhomogeneous} case with advected $\rho$, the Kelvin-Noether theorem in the  first equation in \eqref{HMHD-GMeqns4} becomes 
\cite{HMR1998}.
\begin{align}
\frac{d}{dt} \oint_{c(u)} C 
= \frac{d}{dt} \oint_{c(u)} \mathbf{C}\cdot \dn \bx
= -\,\oint_{c(u)} \frac{1}{\rho}\mb{d} p
\,,\label{tot-KN-them}
\end{align}\end{remark}
and if $\rho$ is constant in \eqref{tot-KN-them} then the circulation of $C$ is conserved.
Hereafter, to simplify notation later, the musical notation $\flat$ and $\sharp$ will be suppressed when no confusion can arise by doing so.

Most of the considerations in the remainder of the paper will involve the \emph{compressible} ideal Hall MHD flow, and some of the conservation laws in \eqref{HMHD-GMeqns4} will also have their analogues for compressible flow. In fact, the geometric notation and variational methods will transfer easily from the compressible to the incompressible case, and also to the standard MHD case when $R=0$.

\subsection{Deterministic Hamilton--Pontryagin variational principle.}
In the geometric formulation of ideal compressible Hall MHD, the natural configuration space to use is the direct product space $G = \operatorname{Diff}_1(\mcal{D})\otimes \operatorname{Diff}_2(\mcal{D})$. Here, the first diffeomorphism group $\operatorname{Diff}_1(\mcal{D})$ models the fluidic part of the dynamics and the second diffeomorphism group $\operatorname{Diff}_2(\mcal{D})$ models the magnetic part of the motion.
For generalarity, we assume the fluidic part of the dynamics consists of two advected quantities, they are the fluid volume density $\rho \in \operatorname{Den}(\mcal{D})$ and the fluid entropy per mass $s \in \mcal{F}(\mcal{D})$. Similarly, the magnetic part of the dynamics also consists of an advected quantity which is the electron charge density $\wt{n}\in\operatorname{Den}(\mcal{D})$. These advected quantities break the symmetry group of the dynamics from the configuration manifold $G$ to isotropy subgroup $\operatorname{Diff}_{\rho_0, s_0}(\mcal{D})\otimes \operatorname{Diff}_{\wt{n}_0}(\mcal{D})$ where the notations are 
\begin{align*}
    \operatorname{Diff}_{\rho_0, s_0}(\mcal{D}) = \{g\in \operatorname{Diff}_1(\mcal{D})\, | \,\rho_0 g = \rho, s_0 g = s_0 \}\,,\quad \text{and}\quad \operatorname{Diff}_{\wt{n}_0}(\mcal{D}) = \{g\in \operatorname{Diff}_2(\mcal{D}) \,|\, \wt{n}_0 g = \wt{n}_0 \}\,.
\end{align*}
To use a symmetry reduced variational principle to derive the equations of motion for ideal Hall MHD, we use the Hamilton--Pontryagin principle 
\begin{align}\label{eqn:minimally_coupled_action_Pontryagin}
\begin{split}
0 = \delta S & =: \delta \int_{t_0}^{t_1} \underbrace{\scp{\mu}{\p_t\wt\phi\cdot\wt\phi^{-1}} - h(\mu, \wt{n})}_{\hbox{Magneto-}}+ \underbrace{\ell(u, \rho, s)+ \scp{\pi}{\p_t\phi\cdot\phi^{-1} - u}}_{\hbox{hydrodynamics}} + \underbrace{\scp{\frac{a\rho}{\wt{n}}\mu}{\p_t\phi\cdot\phi^{-1}}}_{\hbox{minimal coupling}} 
\\
&\qquad\qquad +\underbrace{\scp{\lambda_{\rho}}{\rho_0\cdot\phi^{-1} - \rho_t} + \scp{\lambda_{s}}{s_0\cdot\phi^{-1} - s_t} + \scp{\lambda_{\wt n}}{{\wt n}_0\cdot\wt\phi^{-1} - {\wt n}_t}}_{\hbox{advection contraints}}\,dt\,,
\end{split}
\end{align}
where all variations are arbitrary and constrained to vanish at the endpoints $t = t_0, t_1$. 
The various symbols appearing in the variational principle \eqref{eqn:minimally_coupled_action_Pontryagin} are defined as follows. The fluid velocity is denoted by $u \in \mathfrak{X}_1(\mcal{D})$ which is naturally related to a path in the diffeomorphism group $\phi_t \in \operatorname{Diff}_1(\mcal{D})$ by $u := \p_t \phi\cdot \phi^{-1}$. There are two advected quantities associated with the fluidic part of the motion which are the volume density $\rho \in \operatorname{Den}(\mcal{D})$ and the entropy per unit mass $s \in \mcal{F}(\mcal{D})$.

The magnetic part of the motion is generated by the momentum one-form density, $\mu \in \mathfrak{X}_2^*(\mcal{D})$, which is the dual variable to the vector field 
\[
v:= \p_t \wt{\phi}\cdot\wt{\phi}^{-1} \in \mathfrak{X}_2(\mcal{D})
\]
where $\wt{\phi}\in \operatorname{Diff}_2(\mcal{D})$. Physically, the momentum density $\mu$ is related to the magnetic vector potential one-form $A \in \Lambda^1(\mcal{D})$ by the electron charge density $\wt{n} \in \operatorname{Den}(\mcal{D})$ by the tensor product $R \mu = A\otimes \wt{n}$ where $R \in \mathbb{R}$ is the Larmor ratio (Hall scaling parameter). Furthermore, the charge density $\wt{n}={\wt n}_0\cdot\wt\phi^{-1}$ is advected by the magnetic field part of the dynamics, as
\[
\p_t\nti = -\,\mathcal{L}_{\p_t \wt{\phi}\,\cdot\,\wt{\phi}^{-1}}\nti
=: -\,\mathcal{L}_{v}\nti
\,.
\]
Thus, in the variational principle \eqref{eqn:minimally_coupled_action_Pontryagin}, we have the reduced Lagrangian $\ell(u,\rho, s)$ together with the Hamilton--Pontryagin constraints for the variables $u, \rho$ and $s$ in the fluidic part of the dynamics; as well as a reduced phase space Lagrangian 
\[
\scp{\mu}{\p_t \wt{\phi}\cdot \wt{\phi}^{-1}} - h(\mu, \wt{n})
\,,\] 
together with the push-forward constraints for $\wt{n}$ in the magnetic part of the motion. To model the interaction between the fluidic and magnetic parts of the motion, we introduced a minimal coupling term in the form $\scp{(a\rho/\wt{n})\mu}{u}$ where the constant $a \in \mathbb{R}$ is the electron/ion charge to mass density ratio and $u=\p_t\phi\cdot\phi^{-1}$.

To derive the equations of motion from the variational principle, we vary the action \eqref{eqn:minimally_coupled_action_Pontryagin} to obtain
\begin{align*}
    0 &= \int dt \scp{\frac{\delta\ell}{\delta u} - \pi}{\delta u} + \scp{\frac{\delta\ell}{\delta \rho} + a\left( \,u\intprod \frac{\mu}{\wt{n}}\right)- \lambda_{\rho}}{\delta\rho} + \scp{\frac{\delta\ell}{\delta s}-\lambda_{s}}{\delta s}
    \\
    &\qquad + \scp{\p_t\wt\phi\cdot\wt\phi^{-1}+ \frac{a\rho}{\wt n}u - \frac{\delta h}{\delta\mu}}{\delta\mu} + \scp{-\frac{\delta h}{\delta\wt n}-\frac{a\rho}{\wt n}\left(u\intprod \frac{\mu}{\wt n}\right)-\lambda_{\wt n}}{\delta\wt n} + \scp{\mu}{\delta\left( \p_t\wt\phi\cdot\wt\phi^{-1} \right)}
    \\
    &\qquad + \scp{\delta\pi}{\p_t\phi\cdot\phi^{-1} - u} + \scp{\pi+ \frac{a\rho}{\wt n}\mu}{\delta(\p_t\phi\cdot\phi^{-1})} + \scp{\delta\lambda_{\rho}}{\rho_0\cdot\phi^{-1} - \rho_t} + \scp{\lambda_{\rho}}{\delta(\rho_0\cdot\phi^{-1})}
    \\
    &\qquad + \scp{\delta\lambda_{s}}{s_0\cdot\phi^{-1} - s_t} + \scp{\lambda_{s}}{\delta(s_0\cdot\phi^{-1})} + \scp{\delta\lambda_{\wt n}}{{\wt n}_0\cdot\wt\phi^{-1} - {\wt n}_t} + \scp{\lambda_{\wt n}}{\delta({\wt n}_0\cdot\wt\phi^{-1})} \,. 
\end{align*}
We make use of the following variations, which are the constrained variations on vector fields and advected quantities which correspond to arbitrary curves in ${\rm Diff}(\mcal{D})$ which vanish at the end-points,
\begin{align}
    \delta\left( \p_t\phi\cdot\phi^{-1} \right) &= \p_t\eta - \ad_u\eta \,,\quad\hbox{and}\quad \delta\left(a_0\phi^{-1}\right) = -\mathcal{L}_{\eta}a_t \,,\quad\hbox{where}\quad \eta=\delta\phi\cdot\phi^{-1} \,,\\
    \delta\left( \p_t\wt\phi\cdot\wt\phi^{-1} \right) &= \p_t\wt\eta - \ad_v\wt\eta \,,\quad\hbox{and}\quad \delta\left(\wt a_0\wt\phi^{-1}\right) = -\mathcal{L}_{\wt\eta}\wt a_t \,,\quad\hbox{where}\quad \wt\eta=\delta\wt\phi\cdot\wt\phi^{-1} \,,
\end{align}
and we define the $\diamond$ operator $\diamond: V \times V^* \rightarrow \mathfrak{X}^*(\mcal{D})$ using duality pairings as follows
\begin{align*}
    \scp{\mathcal{L}_ua}{\lambda}_{V^*\times V} 
    = -\scp{\lambda\diamond a}{u}_{\mathfrak{X}^*(\mcal{D})\times\mathfrak{X}(\mcal{D})}\,\quad \forall \lambda\in V \,,\  a\in V^*\,,\ \hbox{and}\  u \in \mathfrak{X}(\mathcal{D}) \,,
\end{align*}
for all vector spaces $V^*$ where $\operatorname{Diff}(\mcal{D})$ admits a right contragradient representation. In our case, we have $V^* = \mcal{F}(\mcal{D})$ and $ V^* = \operatorname{Den}(\mcal{D})$ where the representation is the pullback.

Thus, Hamilton's Principle gives
\begin{align*}
    0 &= \int dt \scp{\frac{\delta\ell}{\delta u} - \pi}{\delta u} + \scp{\frac{\delta\ell}{\delta \rho} + a\left( \,u\intprod \frac{\mu}{\wt{n}}\right)- \lambda_{\rho}}{\delta\rho} + \scp{\frac{\delta\ell}{\delta s}-\lambda_{s}}{\delta s}
    \\
    &\qquad + \scp{\p_t\wt\phi\cdot\wt\phi^{-1}+ \frac{a\rho}{\wt n}u - \frac{\delta h}{\delta\mu}}{\delta\mu} + \scp{-\frac{\delta h}{\delta\wt n}-\frac{a\rho}{\wt n}\left(u\intprod \frac{\mu}{\wt n}\right)-\lambda_{\wt n}}{\delta\wt n} 
    \\
    &\qquad + \scp{-\p_t\mu - \ad^*_v\mu + \lambda_{\wt n}\diamond{\wt n}}{\wt\eta} + \scp{-(\p_t + \ad^*_u)\left(\pi+ \frac{a\rho}{\wt n}\mu\right) +\lambda_{\rho}\diamond\rho+ \lambda_{s}\diamond s}{\eta}
    \\
    &\qquad + \scp{\delta\pi}{\p_t\phi\cdot\phi^{-1} - u} + \scp{\delta\lambda_{\rho}}{\rho_0\cdot\phi^{-1} - \rho_t} + \scp{\delta\lambda_{s}}{s_0\cdot\phi^{-1} - s_t} + \scp{\delta\lambda_{\wt n}}{{\wt n}_0\cdot\wt\phi^{-1} - {\wt n}_t} \,. 
\end{align*}
Assembling the relationships which follow from the fundamental lemma of the calculus of variations, we have
\begin{align}
    \left( \p_t + \ad^*_u \right)\left( \frac{\delta\ell}{\delta u} + \frac{a\rho}{\wt n}\mu \right) &= \left( \frac{\delta\ell}{\delta \rho} + a\left( \,u\intprod \frac{\mu}{\wt{n}}\right)\right)\diamond\rho + \frac{\delta\ell}{\delta s}\diamond s \label{eq:tot momentum eq}
    \,,\\
    \left( \p_t + \ad^*_v \right)\mu &= \left( -\frac{\delta h}{\delta\wt n}-\frac{a\rho}{\wt n}\left(u\intprod \frac{\mu}{\wt n}\right)\right)\diamond\wt n 
    \,,\quad\hbox{where}\quad 
    v = -\frac{a\rho}{\wt n}u + \frac{\delta h}{\delta\mu} \label{eq:mu momentum eq}
    \,,\\
    \left( \p_t + \mathcal{L}_u \right)\rho &= 0 \label{eq:rho advection eq}
    \,,\\
    \left( \p_t + \mathcal{L}_u \right)s &= 0
    \,,\\
    \left( \p_t + \mathcal{L}_v \right)\wt n &= 0 \label{eq:wt n advection eq}
    \,.
\end{align}

\subsection{Ideal compressible Hall MHD}\label{subsec:compressible_example}
Here, we restrict our discussion to the case of ideal compressible Hall MHD in which the reduced phase space Lagrangian for the magnetic dynamics is given by\footnote{The term $|\mb{d}\frac{\mu}{\wt n}|^2$ is defined as follows: $|\mb{d}\frac{\mu}{\wt n}|^2=(\star\mb{d}\frac{\mu}{\wt n})\wedge (\mb{d}\frac{\mu}{\wt n}) = (\star\bs{\delta}\mb{d}\frac{\mu}{\wt n}) \wedge (\frac{\mu}{\wt n}) = \scp{\frac{\mu}{\star\wt n}}{\left( \bs{\delta}\mb{d}\frac{\mu}{\wt n} \right)^\sharp}_{\mathfrak{X}(\mcal{D})^*\times\mathfrak{X}(\mcal{D})}$.} 
\begin{align}
    \scp{\frac{a\rho}{\wt{n}}\mu}{u} - h(\mu, \wt{n}) = \scp{\frac{a\rho}{\wt{n}}\mu}{u} - \frac{R^2}{2}\scp{\frac{\mu}{\star \wt{n}}}{\left(\bs{\delta} \mb{d}\frac{\mu}{\wt{n}}\right)^\sharp} = \int_M a\rho\left( u\intprod \frac{\mu}{\wt{n}}\right) - \frac{R^2}{2}\left|\mb{d}\frac{\mu}{\wt{n}}\right|^2 \,dV\,. \label{eq:Hall MHD phase space Lag}
\end{align}
The functional derivatives of the Hamiltonian $h(\mu, \wt{n})$ are given by 
\begin{align*}
    \frac{\delta h}{\delta \mu} = \frac{R^2}{\star\wt{n}}\left(\bs{\delta} \mb{d}\frac{\mu}{\wt{n}}\right)^\sharp = \frac{R}{\star\wt{n}}\left(\bs{\delta}\mb{d}A\right)^\sharp := v_H\, \quad \text{and}\quad \frac{\delta h}{\delta \wt{n}} = - \frac{R^2}{\star\wt{n}}\left(\bs{\delta} \mb{d}\frac{\mu}{\wt{n}}\right)^\sharp \intprod \frac{\mu}{\wt{n}} = -\frac{1}{\star\wt{n}}\left(\bs{\delta}\mb{d}A\right)^\sharp\intprod A = - \frac{1}{R}\frac{\delta h}{\delta \mu}\intprod A\,.
\end{align*}
where $\bs{\delta}:= \star \mb{d} \star$ is the usual co-differential induced by the metric $g$, the exterior derivative $\mb{d}$ and the Hodge star operator $\star$.
Upon inserting the expression of the functional derivative $\delta h/\delta \mu$ into the definition of the electron velocity vector field $v$, we have the following dynamics for $\wt{n}$,
\begin{align}
    \p_t \wt{n} = -\mcal{L}_v\wt{n} = \mb{d}(v\intprod\wt{n}) = \mb{d}(\star (a\rho u^\flat)) - R \mb{d}(\star\bs{\delta}\mb{d}A) = \mb{d}(\star (a\rho u^\flat)) =  a\mcal{L}_u \rho\,,
\label{chargeNeutral}
\end{align}
where the fourth equality is due to the basic relation $\mb{d}\star\bs{\delta} = 0$. In combination with the advection equation of $\rho$ by $u$ in \eqref{eq:rho advection eq}, one finds 
\begin{align}
    \p_t (a\rho + \wt{n}) = 0 \,.
    \label{chargeNeutral2}
\end{align}
Thus, according to \eqref{chargeNeutral2}, if the algebraic relation $a\rho_0 + \wt{n}_0 = 0$ for local charge neutrality holds initially, then it holds for all time. This is known as the \emph{local charge neutrality condition}. When the local charge neutrality condition is satisfied, we have the simplified electron velocity 
\begin{align}
    v = u + \frac{R}{\star\wt{n}}\left(\bs{\delta}\mb{d}A\right)^\sharp  = u + v_H\,.
\label{v=v+vh}
\end{align}
Next, equations \eqref{eq:mu momentum eq}, \eqref{eq:wt n advection eq} and the definition of magnetic potential in terms of magnetic momentum density, $\mu$, and electron charge density, $\nti$, given by
\begin{align}
\mu/\nti := A/R
\label{eq:mu-def}
\end{align}
imply Ohm's Law
\begin{align}
    (\p_t + \mathcal{L}_v)\frac{\mu}{\wt n} &= \mb{d}\left(-\frac{\delta h}{\delta\wt n} + v \intprod \frac{\mu}{\wt n} - \frac{\delta h}{\delta \mu}\intprod \frac{\mu}{\wt n} \right) = \mb{d}\left( v \intprod \frac{\mu}{\wt n} \right)
    \,,\\
    \hbox{and hence}\qquad \p_t A &= - \mcal{L}_v A + \mb{d}\left(v \intprod A\right) 
    = -v\, \intprod \mb{d} A = -\,v \intprod B \,, \label{eq:ohm'slaw eq}
\end{align}
where the quantity $B := \mb{d}A \in \Lambda^2(\mcal{D})$ denotes the magnetic field two-form (magnetic flux). Further more, taking the exterior derivative $\mb{d}$ gives the advection of the $B$, 
\begin{align}
    \p_t B + \mcal{L}_v B = 0\,.
\label{B-eqn-v}
\end{align}
Physically, this implies that the $B$-field lines are transported by the total velocity $v$ which is the sum of the fluid velocity $u$ representing the ions and the Hall velocity $v_H$ representing the electron current. A direct calculation reveals the ideal Hall MHD equations preserves the magnetic helicity integral $\int_{\mcal{D}} A\wedge B$.
\begin{remark}[The MHD limit]
    In the limit that $R \to 0$, the Hall velocity $v_H \to 0$. Then, the dynamics of $B$ is the advection by the fluid velocity $u$ only
    \begin{align*}
        \p_t B + \mathcal{L}_u B = 0\,,
    \end{align*}
    and the magnetic field lines become frozen into the fluid flow. Thus, when $R\to 0$, variational equations arising from the Hall MHD action principle reduce to the standard equations for compressible MHD flow. We further note that regardless of the value of $R$, the last term in \eqref{eq:Hall MHD phase space Lag} always represents one-half the square of the magnetic field $B:=\mb{d}A$ as occurs in the Euler-Poincar\'e derivation of the standard MHD equations.
\end{remark}

For the fluidic part of the motion, one has the fluid Lagrangian 
\begin{align*}
    \ell(u, \rho, s) = \int_M \frac{1}{2}\rho|u|^2 - \rho\, e(s,\rho)\,dV\,,
\end{align*}
where $|u|^2 := u \intprod u^\flat$, alongside its functional derivatives
\begin{align*}
    \frac{\delta \ell}{\delta u} = \rho u^\flat \,, \quad \frac{\delta\ell}{\delta \rho} = \frac{1}{2} |u|^2 - \wh{h}(p,s) =:{\cal B} \, \quad \hbox{and}\quad \frac{\delta\ell}{\delta s} = - \rho T \,,
\end{align*}
where the First Law of Thermodynamics for the fluid $\mb{d}e = -\,p\mb{d}\rho^{-1}+ T \mb{d}s$ implies that the enthalpy per unit mass $\wh{h}(p,s)$ satisfies
\begin{align} 
    \mb{d}\wh{h} = \rho^{-1}\mb{d}p + T\mb{d}s \quad\hbox{with}\quad \wh{h} = e + \rho \frac{\p e}{\p \rho}\,,
\label{1stLawTds}
\end{align} 
where $p$ denotes thermodynamic pressure and $T$ denotes temperature. Substituting the functional derivatives of the fluid Lagrangian into the momentum equation \eqref{eq:tot momentum eq} and using equation \eqref{eq:ohm'slaw eq} yields the equation for the total momentum 1-form $u^\flat + aA$  as
\begin{align}
    \left(\p_t + \mcal{L}_u\right) \left(u^\flat + \frac{a}{R}A\right) = \mb{d}\left(u\intprod\left(\frac{1}{2}u^\flat + \frac{a}{R}A\right)\right) - \frac{1}{\rho}\mb{d}p\,.
\end{align}
Notice that, from equation \eqref{eq:ohm'slaw eq}, we have, upon dropping $\flat$ to simplify notation because $A$ is understood to be a 1-form,
\begin{equation*}
    (\p_t + \mathcal{L}_u)A = \mb{d}(v\intprod A) + \mathcal{L}_{u-v}A = \mb{d}(u\intprod A) + (u-v)\intprod \mb{d}A\,,
\end{equation*}
and hence the the equation for the fluid momentum $u^\flat$ is
\begin{align}
\begin{split}
    \left(\p_t + \mcal{L}_u\right)u^\flat &= \frac{1}{2}\mb{d}\left(u\intprod u^\flat\right) - \frac{1}{\star\rho}\mb{d}p + (v-u)\intprod \frac{a}{R}\mb{d}A \\
    &= \frac{1}{2}\mb{d}\left(u\intprod u^\flat\right) - \frac{1}{\star\rho}\left(\mb{d}p + (\bs{\delta}B)^\sharp \intprod B\right)\,,
\end{split}
\end{align}
where in the last equality we have assumed the local charge neutrality condition. Thus, by assuming that the initial conditions $\rho_0$ and $\wt{n}_0$ satisfy $a\rho_0 = -\wt{n}_0$, one simplifies the equation set \eqref{eq:tot momentum eq} -- \eqref{eq:wt n advection eq} to 
\begin{align}
\begin{split}
    &\left(\p_t + \mcal{L}_u\right)u^\flat = \frac{1}{2}\mb{d}\left(u\intprod u^\flat\right) - \frac{1}{\star\rho}\left(\mb{d}p(e,s) + (\bs{\delta}B)^\sharp \intprod B\right)\,,\\
    &\p_t A + u \intprod B - \frac{R}{a(\star\rho)}(\bs{\delta}B)^\sharp \intprod B = 0\,,\\
    &\left( \p_t + \mathcal{L}_u \right)\rho = 0
    \,,\\
    &\left( \p_t + \mathcal{L}_u \right)s = 0 \,, \quad \text{with} \quad a\rho = -\wt{n}\,.
\end{split}
\end{align}
\begin{remark}[Ideal Compressible Hall MHD in 3D]
When working in a three dimensional domain with the Euclidean metric, we have the following definition of the Hall velocity $\bv_H$ expressed in terms of basis functions
\begin{align*}
\begin{split}
v_H = \bv_H\cdot \nabla = \frac{\delta h}{\delta \mu}
= \frac{R}{\star\wt{n}}\left(\bs{\delta}\mb{d}A\right)^\sharp 
=  \frac{R}{\star\wt{n}}\left(\star \mb{d} \star\mb{d}A\right)^\sharp
&= \frac{R}{\star\wt{n}}\left(\left(\mathrm{curl}\mathrm{curl}\bA\right)\cdot d\bx\right)^\sharp
\\
&= \frac{R}{\star\wt{n}}\left(\mathrm{curl}\mathrm{curl}\bs{A}\right)\cdot \nabla
= \frac{R}{\star\wt{n}}\left(\mathrm{curl}\bs{B}\right)\cdot \nabla
\,,
\end{split}
\end{align*}
where $\mb{A}$ and $\mb{B}$ are the basis coefficients of the magnetic potential one-form $A:=\mb{A}\cdot d\mb{x}$ and the magnetic field two-form $B:=\mb{d}A=\mathrm{curl}\mb{A}\cdot \star d\mb{x}$ respectively.
Then, by \eqref{eq:mu momentum eq} we have
\begin{align*}
v = -\frac{a\rho}{\wt n}u + \frac{\delta h}{\delta\mu} 
= \left(\bs{u} - \frac{R}{a\rho}\mathrm{curl}\bs{B}\right)\cdot \nabla
\,.
\end{align*}

Following from the previous specialisation to the Euclidean 3D domain, we may write the ideal compressible Hall MHD equations in vector calculus form as
\begin{align}
    \begin{split}
        &\p_t \bu + \bu\cdot\nabla \bu + \frac{1}{\rho}\left(\nabla p(e,s) + \mb{B} \times \operatorname{curl}\mb{B}\right) = 0\,,\\
        &\p_t \mb{A} + \mb{B}\times 
        \left(\bu - \frac{R}{a\rho}\operatorname{curl}\mb{B}\right) = 0\,,\\
        &\p_t \rho + \nabla\cdot(\bu\rho) = 0\,,\\
        &\p_t s + \bu\cdot\nabla s = 0\,,
    \end{split}
\end{align}
When dissipation by viscosity and resistivity are included in these equations, they become globally well-posed \cite{HanHuLai2024}. 

For isentropic (or isothermal) Hall MHD flow, the term $T\mb{d}s$ vanishes  in \eqref{1stLawTds} and the Hall MHD equations preserve the \emph{Hall helicity} defined as the spatial integral over the domain
\begin{align*}
\Lambda_H=\int C\wedge \mb{d} C = \int \mb{C}\cdot \mathrm{curl}\mb{C} \,d^3x
\quad\hbox{with}\quad 
\mathbf{C} := \mathbf{u} + \frac{a}{R} \mathbf{A}
\end{align*}
as defined in equation \eqref{tot-circ-vector}.
The Hall helicity $\Lambda_H$ is a Casimir functional for the Lie-Poisson bracket.
Conservation of the helicity $\Lambda_H$ implies preservation by isentropic Hall MHD of the linkage number of the field lines of $\mathrm{curl}\mb{C}$. Moreover, a critical point of the sum of the Hamiltonian and the Hall helicity is an equilibrium flow. 

Cases with $\mathrm{curl}\mb{C}=\lambda \mb{C}$ with constant $\lambda\ne 0$ can exhibit Lagrangian chaos as an analogue of the ABC flow of MHD \cite{AK2021}. Vanishing helicity results in a Hall Beltrami flow for  which $\mb{C}\times\mathrm{curl}\mb{C}= \nabla \phi$ and the evolution of $\mb{C}$ takes place on level sets of $\phi$. For the corresponding Euler fluid case, see, e.g., \cite{HolmKim1991}.
\end{remark}

\subsection{Deterministic Hamiltonian structure}\label{subsec:Hamiltonian_deterministic}

Returning to Hamilton's Principle as applied in equation \eqref{eqn:minimally_coupled_action_Pontryagin}, we may Legendre transform the fluid part of the dynamics to obtain a Hamiltonian, $\mathfrak{h}(\pi,\rho,s)$, defined as
\begin{equation}\label{eqn:fluid_Legendre_transform}
    \ell(u,\rho,s) := \scp{\pi}{u} - \mathfrak{h}(\pi,\rho,s) \,.
\end{equation}
Substituting this into equation \eqref{eqn:minimally_coupled_action_Pontryagin}, we have
\begin{align}\label{eqn:minimally_coupled_action_Pontryagin_Hamiltonian1}
\begin{split}
0 = \delta S & =: \delta \int_{t_0}^{t_1} \scp{\mu}{\p_t\wt\phi\cdot\wt\phi^{-1}} - h(\mu, \wt{n})+ \scp{\pi}{\p_t\phi\cdot\phi^{-1}} - \mathfrak{h}(\pi,\rho,s) + \scp{\frac{a\rho}{\wt{n}}\mu}{\p_t\phi\cdot\phi^{-1}} 
\\
&\qquad\qquad +\scp{\lambda_{\rho}}{\rho_0\cdot\phi^{-1} - \rho_t} + \scp{\lambda_{s}}{s_0\cdot\phi^{-1} - s_t} + \scp{\lambda_{\wt n}}{{\wt n}_0\cdot\wt\phi^{-1} - {\wt n}_t}\,dt\,.
\end{split}
\end{align}
This form naturally lends itself to the identification of a new variable, $M:=\pi + \frac{a\rho}{\wt n}\mu$, and a Hamiltonian defined as the sum of the magnetic and fluid Hamiltonians as
\begin{equation}\label{eqn:total_Hamiltonian_definition}
    \mcal{H}(M,\rho,s,\mu,\wt n) := h(\mu,\wt n) + \mathfrak{h}\left(M-\frac{a\rho}{\wt n}\mu , \rho, s \right) \,.
\end{equation}
Hamilton's Principle is therefore
\begin{align}\label{eqn:minimally_coupled_action_Pontryagin_Hamiltonian2}
\begin{split}
0 &= \delta S =: \delta \int_{t_0}^{t_1} \scp{\mu}{\p_t\wt\phi\cdot\wt\phi^{-1}} + \scp{M}{\p_t\phi\cdot\phi^{-1}} - \mcal{H}(M,\rho,s,\mu,\wt n) 
\\
&\qquad\qquad +\scp{\lambda_{\rho}}{\rho_0\cdot\phi^{-1} - \rho_t} + \scp{\lambda_{s}}{s_0\cdot\phi^{-1} - s_t} + \scp{\lambda_{\wt n}}{{\wt n}_0\cdot\wt\phi^{-1} - {\wt n}_t}\,dt
\\
&=\int_{t_0}^{t_1} \scp{\delta\mu}{\p_t\wt\phi\cdot\wt\phi^{-1} - \frac{\delta \mcal{H}}{\delta\mu}} + \scp{\delta M}{\p_t\phi\cdot\phi^{-1} - \frac{\delta \mcal{H}}{\delta M}} + \scp{-(\p_t+\ad^*_{\p_t\wt\phi\cdot\wt\phi^{-1}})\mu + \lambda_{\wt n}\diamond\wt n}{\wt\eta}
\\
&\qquad\qquad +\scp{-(\p_t+\ad^*_{\p_t\phi\cdot\phi^{-1}})M + \lambda_{\rho}\diamond\rho + \lambda_s\diamond s}{\eta} - \scp{\lambda_{\rho} + \frac{\delta \mcal{H}}{\delta \rho}}{\delta\rho} -\scp{\lambda_s + \frac{\delta \mcal{H}}{\delta s}}{\delta s}
\\
&\qquad\qquad -\scp{\lambda_{\wt n} + \frac{\delta \mcal{H}}{\delta \wt n}}{\delta\wt n} +\scp{\delta\lambda_{\rho}}{\rho_0\cdot\phi^{-1} - \rho_t} + \scp{\delta\lambda_{s}}{s_0\cdot\phi^{-1} - s_t} + \scp{\delta\lambda_{\wt n}}{{\wt n}_0\cdot\wt\phi^{-1} - {\wt n}_t}\,dt \,,
\end{split}
\end{align}
The equations corresponding to this have the following Lie-Poisson structure, expressed here in matrix form
\begin{equation}\label{eqn:Hall MHD PB}
    \p_t
    \begin{pmatrix}
        M \\ \rho \\ s \\ \mu \\ \wt n
    \end{pmatrix}
    =
    -
    \begin{pmatrix}
        \ad^*_\Box M & \Box \diamond \rho & \Box \diamond s & 0 & 0
        \\
        \mathcal{L}_\Box\rho & 0 & 0 & 0 & 0
        \\
        \mathcal{L}_\Box s & 0 & 0 & 0 & 0
        \\
        0 & 0 & 0 & \ad^*_\Box \mu & \Box \diamond \wt n
        \\
        0 & 0 & 0 & \mathcal{L}_\Box \wt n & 0
    \end{pmatrix}
    \begin{pmatrix}
        \delta\mcal{H} / \delta M \\ \delta\mcal{H} / \delta \rho \\ \delta\mcal{H} / \delta s \\ \delta\mcal{H} / \delta\mu \\ \delta\mcal{H} / \delta \wt n
    \end{pmatrix}
    \,.
\end{equation}
The restult of this is that the examples considered in this paper have a Lie-Poisson structure, which we illustrate below for three dimensional compressible Hall MHD on a Euclidean domain. For this example, we will use vector calculus notation to better illustrate the problem for readers less familiar with exterior calculus. When division by a density occurs in this section, it is not necessary to take the same precaution as discussed in Remark \ref{remark:dividing_by_density}. Instead, we always mean division by the scalar dual to the density. 

The Hamiltonian corresponding to the model of compressible MHD presented in Section \ref{subsec:compressible_example} is
    \begin{align}\label{eqn:compressible_Hamiltonian}
        {\mcal{H}}(M, \rho, s, \mu, \wt{n}) = \int_{\mcal{D}}\frac{1}{2\rho}\left|\bs{M} - \frac{a\rho}{\wt{n}}\bs{\mu}\right|^2 + \rho e(\rho, s) + \frac{R^2}{2}\left|\operatorname{curl}\frac{\bs{\mu}}{\wt{n}}\right|^2 d^3x\,.
    \end{align}
The Poisson bracket, in vector calculus notation, is
    \begin{align}\label{eqn:LPB-3D-euclidean}
         \p_t \begin{pmatrix}
            M_i \\ \rho \\ s \\ \mu_i \\ \wt{n}
        \end{pmatrix}
        = -
        \begin{pmatrix}
            M_i\p_j + \p_iM_j & \rho \p_i & -s,_j & 0 & 0 \\
            \p_j \rho & 0& 0& 0& 0\\
            s,_j &0 &0 &0 &0 \\
            0 &0 &0 & \mu_i\p_j + \p_i\mu_j& \wt{n}\p_i \\
            0 &0 &0 & \p_j\wt{n} & 0
        \end{pmatrix}
        \begin{pmatrix}
            \delta \mcal{H} / \delta M_j \\
            \delta \mcal{H} / \delta \rho \\
            \delta \mcal{H} / \delta s \\
            \delta \mcal{H} / \delta \mu_j \\
            \delta \mcal{H} / \delta \wt{n} 
        \end{pmatrix}\,,
    \end{align}
which is the Lie-Poisson bracket on the following Lie co-algebra:
    \begin{align}\label{eqn:Hall MHD Lie co algebra}
        \mathfrak{X}^*(\mcal{D}) \ltimes \left(\Lambda^0(\mcal{D})\oplus \Lambda^n(\mcal{D})\right)\oplus \left(\mathfrak{X}^*(\mcal{D})\ltimes \Lambda^n(\mcal{D})\right)\,.
    \end{align}
Varying the Hamiltonian yields,
    \begin{align}\label{eqn:Hamiltonian_variations}
    \begin{split}
        \delta {\mcal{H}} &= \int_{\mcal{D}} \delta \bs{M}\cdot \frac{1}{\rho}\left(\bs{M} - \frac{a\rho}{\wt{n}}\bs{\mu}\right) + \delta \bs{\mu}\cdot \left(\frac{1}{\rho}\left(\bs{M} - \frac{a\rho}{\wt{n}}\bs{\mu}\right)\left(-\frac{a\rho}{\wt{n}}\right) + \frac{1}{\wt{n}}\operatorname{curl}\operatorname{curl}\frac{\bs{\mu}}{\wt{n}}\right) 
        \\
        & \qquad + \delta \wt{n} \left(\frac{1}{\rho}\left(\bs{M} - \frac{a\rho}{\wt{n}}\bs{\mu}\right)\cdot\left(\frac{a\rho}{\wt{n}^2}\bs{\mu}\right) - 
        \left(\operatorname{curl}\operatorname{curl}\frac{\bs{\mu}}{\wt{n}} \right)
        \cdot \frac{\bs{\mu}}{\wt{n}^2} \right) + \delta s\,\rho T 
        \\
        & \qquad + \delta\rho\left( \wh{h}(p,s) - \frac{1}{2\rho^2}\Big| \bs{M} - \frac{a\rho}{\wt n}\bs{\mu}\Big|^2 - \frac{a\bs{\mu}}{\rho\wt n}\Big(\bs{M} - \frac{a\rho}{\wt n}\bs{\mu} \Big) \right) d^3x
    \end{split}
    \end{align}
and assembling the variational derivatives into the above Lie-Poisson equations \eqref{eqn:LPB-3D-euclidean} results in the expected equations of motion also derived in Section \ref{subsec:compressible_example}. In particular, expressed 
    in terms of the vector fields
    \begin{align*}
        \bs{u} := \frac{1}{\rho}\left(\bs{M} - \frac{a\rho\bs{\mu}}{\wt{n}}\right)
        \,,\quad \text{and}\quad \bs{v} := \frac{1}{\wt{n}}\operatorname{curl}\operatorname{curl}\frac{\bs{\mu}}{\wt{n}}\,,
    \end{align*} 
    the following equations arise for the magnetic variables,
    \begin{align}
        \begin{split}
            &\p_t \mu + \ad^*_{-\frac{a\rho}{\wt{n}}u} \mu + \ad^*_{v} \mu +\wt{n}\mb{d}\left(\bs{u}\cdot\frac{a\rho}{\wt{n}}\bs{\mu} - \frac{\bs{\mu}}{\wt{n}^2}\cdot \bs{v}\right) = 0 \,,\\
            &\p_t \wt{n} + \mcal{L}_{-\frac{a\rho}{\wt{n}}u} \wt{n} + \mcal{L}_{v} \wt{n} = 0 \,.
        \end{split}
    \end{align}

\begin{remark}[An alternative Hamiltonian structure]
    Making the change of variables
    \begin{align}
        \left(M, \rho, s , \mu, \wt{n}\right)\longrightarrow \left(M, \rho, s , A, \wt{n}\right) 
    \,,\end{align}
    where $\mu := A\otimes \wt{n} \in \mathfrak{X}^*(\mcal{D})$ is expressed in terms of the magnetic vector potential, $A = \bs{A}\cdot d\bx \in \Lambda^1(\mcal{D})$, and the fluid charge density $\wt{n} d^3x \in \Lambda^n(\mcal{D})$. The relevant Lie co-algebra becomes $\mathfrak{X}^*(\mcal{D}) \ltimes \left(\Lambda^0(\mcal{D})\oplus \Lambda^n(\mcal{D})\right)\otimes \left(\Lambda^1(\mcal{{D}})\oplus \Lambda^n(\mcal{D})\right)$ and the Poisson bracket transforms into
    \begin{align}
        \p_t \begin{pmatrix}
            M_i \\ \rho \\ s \\ A_i \\ \wt{n}
        \end{pmatrix}
        = -
        \begin{pmatrix}
            M_i\p_j + p_iM_j & \rho \p_i & -s,_j & 0 & 0 \\
            \p_j \rho & 0& 0& 0& 0\\
            s,_j &0 &0 &0 &0 \\
            0 &0 &0 &\wt{n}^{-1}\left(A_i, _j - A_j,_i\right) & \p_i \\
            0 &0 &0 & \p_j & 0
        \end{pmatrix}
        \begin{pmatrix}
            \delta \wt{\mcal{H}} / \delta M_j\\
            \delta \wt{\mcal{H}} / \delta \rho \\
            \delta \wt{\mcal{H}} / \delta s \\
            \delta \wt{\mcal{H}} / \delta \mu_j \\
            \delta \wt{\mcal{H}} / \delta \wt{n} 
        \end{pmatrix}\,,
    \label{HMHDsystem}
    \end{align}
    where the Hamiltonian, ${\cal H}$, has been expressed as a Hamiltonian, $\wt{\mcal{H}}$, in the new variables as
    \begin{align}
        \wt{\mcal{H}}(M, \rho, s, A, \wt{n}) = \int_{\mcal{D}}\frac{1}{2\rho}\left|\bs{M} - \frac{a\rho}{R}\bs{A}\right|^2 + \rho e(\rho, s) + \frac{1}{2}\left|\operatorname{curl}\bs{A}\right|^2 d^3x\,.
    \end{align}
    This Hamiltonian formulation of Hall MHD agrees with the existing literature on the topic \cite{Holm1987,HK1987}.
\end{remark}

\subsection{Ideal incompressible Hall MHD}\label{subsec:incompressible_example}

The geometric structures presented thus far in this section have been tailored for compressible Hall MHD. Here, we will show that only a minimal modification is required to consider the inhomogeneous incompressible case. In this case, the collection of quantities which are advected by the velocity $u$ becomes the mass density, $D \in {\rm Den}(\mathcal{D})$, and the potential temperature, $\rho \in \Lambda^0(\mathcal{D})$. As before, we also have the charge density, $\wt n \in {\rm Den}(\mathcal{D})$, which is advected by the velocity $v = u+v_H$. Incompressible Hall MHD can be derived by modifying Hamilton's Principle as it appears in equation \eqref{eqn:minimally_coupled_action_Pontryagin} to reflect the fact that the advected quantities are different and to alter the minimal coupling term to be $\scp{\frac{D\rho a}{\wt n}\mu}{\p_t\phi\cdot\phi^{-1}}$. Note that this is simply a consequence of the fact that the volume measure has changed to be $D\rho$. Once these changes are made, the Hamiltonian structure is qualitatively unchanged. The fluid part of the Hamiltonian is now taken to correspond to inhomogeneous thermal Euler, and the total Hamiltonian for divergence free Hall MHD is now taken to be
\begin{equation}
    {\mcal{H}}_{\rm div}(M, D, \rho, \mu, \wt{n}) = \int_{\mcal{D}}\frac{1}{2D\rho}\left|M - \frac{D\rho a}{\wt{n}}\mu\right|^2 + p(D-1) + \frac{R^2}{2}\left|\mb{d}\frac{\mu}{\wt{n}}\right|^2 dV \,,
\end{equation}
taken together with the Lie-Poisson structure
\begin{equation}\label{eqn:Hall MHD PB incompressible}
    \p_t
    \begin{pmatrix}
        M \\ D \\ \rho \\ \mu \\ \wt n
    \end{pmatrix}
    =
    -
    \begin{pmatrix}
        \ad^*_\Box M & \Box \diamond D & \Box \diamond \rho & 0 & 0
        \\
        \mathcal{L}_\Box\rho & 0 & 0 & 0 & 0
        \\
        \mathcal{L}_\Box s & 0 & 0 & 0 & 0
        \\
        0 & 0 & 0 & \ad^*_\Box \mu & \Box \diamond \wt n
        \\
        0 & 0 & 0 & \mathcal{L}_\Box \wt n & 0
    \end{pmatrix}
    \begin{pmatrix}
        \delta\mcal{H}_{\rm div} / \delta M \\ \delta\mcal{H}_{\rm div} / \delta D \\ \delta\mcal{H}_{\rm div} / \delta \rho \\ \delta\mcal{H}_{\rm div} / \delta\mu \\ \delta\mcal{H}_{\rm div} / \delta \wt n
    \end{pmatrix}
    \,.
\end{equation}
Computing the variational derivatives and substituting these into the Lie-Poisson equations, we get
\begin{align}
    (\p_t + \ad^*_u)M &= -D\mb{d}\left( p - \frac{\rho|u|^2}{2} - a\rho\,u\intprod\frac{\mu}{\wt n} \right) + \left(- \frac{D|u|^2}{2} - Da\,u\intprod \frac{\mu}{\wt n}\right)\mb{d}\rho
    \,,\\
    \hbox{where}\quad  M &= D\rho u^\flat + \frac{D\rho a}{\wt n}\mu
    \,,\\
    (\p_t + \ad^*_v)\mu &= \wt n \mb{d}\left( v \intprod\frac{\mu}{\wt n} \right)
    \,,\\
    \hbox{where}\quad v &= -\frac{D\rho a}{\wt n}u + \frac{R^2}{\star\wt n}\left(\bs{\delta}\mb{d}\frac{\mu}{\wt n} \right)^\sharp = -\frac{D\rho a}{\wt n}u + v_H
    \,,
\end{align}
together with the expected advection equations for $D$, $\rho$, and $\wt n$. Considering that pressure enforces $D=1$ and local charge neutrality holds, $a\rho + \star \wt n = 0$, the equation for $\mu/\wt n$ can be written as
\begin{equation}
\begin{aligned}
    (\p_t + \mcal{L}_u)\frac{\mu}{\wt n} &= -\mathcal{L}_{\frac{R^2}{\star\wt n}\left( \bs{\delta}\mb{d}\frac{\mu}{\wt n} \right)^\sharp}\frac{\mu}{\wt n} - \mb{d}\left( -\frac{R^2}{\star\wt n}\left(\bs{\delta}\mb{d}\frac{\mu}{\wt n} \right)^\sharp \intprod \frac{\mu}{\wt n} + \frac{D\rho a}{\wt n}u\intprod\frac{\mu}{\wt n} \right) 
    \\
    &= \mb{d}\left( u\intprod\frac{\mu}{\wt n} \right) -\frac{R^2}{\star\wt n}\left(\bs{\delta}\mb{d}\frac{\mu}{\wt n} \right)^\sharp \intprod \mb{d}\frac{\mu}{\wt n} \,,
\end{aligned}
\end{equation}
and the equation for the fluid velocity $u^\flat$ is
\begin{equation}
\begin{aligned}
    \left( \p_t + \mcal{L}_u \right)u^\flat &= -\frac{1}{\rho}\mb{d}p + \mb{d}\left( \frac{|u|^2}{2} \right) +\frac{R^2a}{\star\wt n}\left(\bs{\delta}\mb{d}\frac{\mu}{\wt n} \right)^\sharp \intprod \mb{d}\frac{\mu}{\wt n} \,,\quad\hbox{where}\quad \operatorname{div}u=0 \\
    & = -\frac{1}{\rho} \mb{d}p + \mb{d}\left(\frac{|u|^2}{2}\right) - \frac{1}{\rho} (\bs{\delta}\mb{d}A)^\sharp \intprod \mb{d}A\,.
\end{aligned}
\end{equation}
Notice now that the equations derived in this section are precisely equations \eqref{HMHD-GMeqns} introduced at the beginning of this paper.

\section{Stochastic Hamilton-Pontryagin principles for Hall MHD}\label{sec-EulPoinc}
\subsection{Stochastic variational principle}\label{subsec:stochastic_variational_principle}
We will modify the above Hamilton-Pontryagin principle to include stochastic terms following Holm \cite{Holm2015}. The variations will be constructed formally as in \cite{ST2023}, wherein it is demonstrated that the stochastic variational procedure is constructed such that it is rigorously defined and preserves many important properties from geometric mechanics.

The variations are constructed in the following way. Fixing arbitrary vector fields $\eta,\widetilde\eta \in \mathfrak{X}(\mcal{D})$ which are once differentiable in time, we construct group-valued perturbations $e_{\epsilon,t},\widetilde e_{\epsilon,t} \in {\rm Diff}(\mcal{D})$ as solutions to the following equations
\begin{equation}
    \p_t e_{\epsilon,t} = \epsilon(\p_t\eta)\cdot e_{\epsilon,t} \,,\quad\hbox{and}\quad \p_t\widetilde e_{\epsilon,t} = \epsilon(\p_t\widetilde\eta)\cdot \widetilde e_{\epsilon,t}\,,
\end{equation}
where $(\,\cdot\,)$ denotes the lifted right action of group elements on the algebra. These may be used to construct a family of variations of the curves $\phi_t$ and $\widetilde\phi_t$ as
\begin{equation}
    \phi_{\epsilon,t} = e_{\epsilon,t}\cdot\phi_t \,,\quad\hbox{and}\quad \widetilde\phi_{\epsilon,t} = \widetilde e_{\epsilon,t}\cdot\widetilde\phi_t \,.
\end{equation}
As was shown in \cite{ACC2014}, this construction mimics the exponential map since
\begin{equation*}
    \frac{\p}{\p\epsilon}\bigg|_{\epsilon=0}e_{\epsilon,t} = \eta \,,\quad \frac{\p}{\p\epsilon}\bigg|_{\epsilon=0}e_{\epsilon,t}^{-1} = -\eta \,,\quad \frac{\p}{\p\epsilon}\bigg|_{\epsilon=0}\widetilde e_{\epsilon,t} = \widetilde\eta \,,\quad\hbox{and}\quad \frac{\p}{\p\epsilon}\bigg|_{\epsilon=0}\widetilde e_{\epsilon,t}^{-1} = -\widetilde\eta \,,
\end{equation*}
which implies that
\begin{equation*}
    \delta\phi\cdot\phi^{-1} = \left(\frac{\p}{\p\epsilon}\bigg|_{\epsilon=0}\phi_{\epsilon,t}\right)\cdot\phi^{-1} = \eta \,,\quad\hbox{and}\quad \delta\widetilde\phi\cdot\widetilde\phi^{-1} = \widetilde\eta \,.
\end{equation*}
This construction results in the following formulas for the variations
\begin{equation}\label{eqn:stochastic_Lin_constraints}
    \delta(\diff \phi \cdot \phi^{-1}) = \p_t\eta \,\diff t - \ad_{\diff\phi\cdot\phi^{-1}}\eta \,,\quad\hbox{and}\quad \delta(\diff \widetilde\phi \cdot \widetilde\phi^{-1}) = \p_t{\widetilde\eta} \,\diff t - \ad_{\diff\widetilde\phi\cdot\widetilde\phi^{-1}}\widetilde\eta \,,
\end{equation}
which are proven in \cite{ST2023}, where the precise meaning of group actions is described. Furthermore, by using the same construction and an analogous proof, it follows that
\begin{equation}\label{eqn:stochastic_Lin_constraints_advected_quantities}
    \delta\rho = \frac{\p}{\p\epsilon}\bigg|_{\epsilon=0}(\rho_0\cdot\phi_{\epsilon,t}^{-1}) = -\mathcal{L}_{\eta}\rho \,,\quad\hbox{and, similarly,}\quad \delta s = -\mathcal{L}_\eta s \,,\quad \delta\widetilde n = -\mathcal{L}_{\widetilde\eta}\widetilde n \,,
\end{equation}
where $\rho_t:=\rho_0\cdot\phi_{\epsilon,t}^{-1}$ denotes the push-forward of the advected quantity $\rho_0$ by the flow map $\phi_{\epsilon, t}$.

Following \cite{SC2021,ST2023}, it remains to give meaning to the notation $\diff \phi\cdot \phi^{-1}$ and $\diff \widetilde\phi\cdot \widetilde\phi^{-1}$ by describing how they are compatible with a driving semimartingale $(t,W_t^1,W_t^2,\dots ,B_t^1,B_t^2,\dots)$, where $W_t^i$ and $B_t^i$ are two collections of i.i.d. Brownian motions. Here, we will integrate with respect to $W_t^i$ to incorporate noise into the fluid part of the physics, and $B_t^i$ will correspond to the magnetic dynamics. Note that, in practice, it may be convenient to choose $W_t$ and $B_t$ to be the same. However, the distinction is made here to keep track of which terms in the structure correspond to the magnetic physics and which to the fluid. To insert transport-type noise into the fluid part of the motion, the Hamilton-Pontryagin constraint will enforce that 
\[
\diff\phi\cdot\phi^{-1} = u\,dt + \sum_i \xi_i\circ \diff W_t^i\,,
\] 
where $\xi_i$ are given fixed vector fields on $M$. On the Hamiltonian side, this corresponds to a symplectic diffusion with additional stochastic Hamiltonians given by $\scp{m}{\xi_i}$ where $m$ is the momentum of the fluid. It is mathematically necessary to assume that $\diff\widetilde\phi\cdot\widetilde\phi^{-1}$ is compatible with the same semimartingale. Since the system corresponding to the curve of diffeomorphisms $\widetilde\phi_t$ is to be described by a reduced phase space Lagrangian, we will enforce this compatibility by including stochastic Hamiltonians $h_i(\mu,\wt n)$. Finally, the minimal coupling term will be augmented with the stochastic transport terms from the fluid, to ensure that the entire stochastic flow is coupled rather than just the drift part. The stochastic variational principle is then the following,
\begin{align}\label{eqn:minimally_coupled_action_Pontryagin_stochastic}
\begin{split}
0 = \delta S & =: \delta \int \scp{\mu}{\diff\wt\phi\cdot\wt\phi^{-1}} - h(\mu, \wt{n})\,\diff t - \sum_i h_i(\mu,\wt n)\circ\diff B_t^i + \ell(u, \rho, s) \,\diff t
\\
&\qquad\qquad + \scp{\pi}{\diff\phi\cdot\phi^{-1} - u\,\diff t - \sum_i\xi_i\circ \diff W_t^i} + \scp{\frac{a\rho}{\wt{n}}\mu}{\diff\phi\cdot\phi^{-1}}
\\
&\qquad\qquad +\scp{\diff\lambda_{\rho}}{\rho_0\cdot\phi^{-1} - \rho_t} + \scp{\diff\lambda_{s}}{s_0\cdot\phi^{-1} - s_t} + \scp{\diff\lambda_{\wt n}}{{\wt n}_0\cdot\wt\phi^{-1} - {\wt n}_t}\,,
\end{split}
\end{align}
where, following \cite{SC2021}, the Hamilton-Pontryagin constraints for the advected quantities in the final line have been enforced with stochastic Lagrange multipliers of the form $\diff\lambda$, where $\diff$ denotes the fact that $\lambda$ is assumed to be compatible with the semimartingale $(t,W_t^1,W_t^2,\dots ,B_t^1,B_t^2,\dots)$.

Varying the action in an analogous manner to the deterministic case, using the stochastic constrained variations \eqref{eqn:stochastic_Lin_constraints} and \eqref{eqn:stochastic_Lin_constraints_advected_quantities}, we have
\begin{align*}
    0 &= \int \scp{\frac{\delta\ell}{\delta u} - \pi}{\delta u}\,\diff t + \scp{\frac{\delta\ell}{\delta \rho}\,\diff t + a (\diff\phi\cdot\phi^{-1})\intprod \frac{\mu}{\wt{n}} - \diff\lambda_{\rho}}{\delta\rho}
    \\
    &\qquad + \scp{\frac{\delta\ell}{\delta s}\,\diff t-\diff\lambda_{s}}{\delta s} + \scp{\diff\wt\phi\cdot\wt\phi^{-1}- \frac{\delta h}{\delta\mu}\,\diff t - \sum_i\frac{\delta h_i}{\delta\mu}\circ\diff B_t^i + \frac{a\rho}{\wt n}\diff\phi\cdot\phi^{-1}}{\delta\mu} 
    \\
    &\qquad + \scp{-\frac{\delta h}{\delta\wt n}\,\diff t -\sum_i\frac{\delta h_i}{\delta\wt n} \circ\diff B_t^i -\frac{a\rho}{\wt n}(\diff\phi\cdot\phi^{-1})\intprod \frac{\mu}{\wt n} -\diff \lambda_{\wt n}}{\delta\wt n} 
    \\
    &\qquad + \scp{-\left(\diff + \ad^*_{\diff\wt\phi\cdot\wt\phi^{-1}}\right)\mu + \diff\lambda_{\wt n}\diamond{\wt n}}{\wt\eta} + \scp{-\left(\diff + \ad^*_{\diff\phi\cdot\phi^{-1}}\right)\left(\pi + \frac{a\rho}{\wt n}\mu\right) +\diff\lambda_{\rho}\diamond\rho + \diff\lambda_{s}\diamond s}{\eta}
    \\
    &\qquad + \scp{\delta\pi}{\diff\phi\cdot\phi^{-1} - u\,\diff t - \sum_i\xi_i\circ \diff W_t^i} + \scp{\diff\delta\lambda_{\rho}}{\rho_0\cdot\phi^{-1} - \rho_t}
    \\
    &\qquad + \scp{\diff\delta\lambda_{s}}{s_0\cdot\phi^{-1} - s_t} + \scp{\diff\delta\lambda_{\wt n}}{{\wt n}_0\cdot\wt\phi^{-1} - {\wt n}_t} \,. 
\end{align*}
The fundamental lemma of the stochastic calculus of variations \cite{ST2023} implies the following system of SPDEs expressed in their geometric form
\begin{align}
    \left( \diff + \ad^*_{\diff\phi\cdot\phi^{-1}} \right)\left( \frac{\delta\ell}{\delta u} + \frac{a\rho}{\wt n}\mu \right) &= \left( \frac{\delta\ell}{\delta \rho}\,\diff t + a \,(\diff\phi\cdot\phi^{-1})\intprod \frac{\mu}{\wt{n}}\right)\diamond\rho + \frac{\delta\ell}{\delta s}\diamond s\,\diff t
    \,,\\
    \hbox{where}\quad \diff\phi\cdot\phi^{-1} &= u\,\diff t + \sum_i\xi_i\circ \diff W_t^i
    \,,\\
    \left( \diff + \ad^*_{\diff\wt\phi\cdot\wt\phi^{-1}} \right)\mu &= \left( -\frac{\delta h}{\delta\wt n}\,\diff t-\frac{a\rho}{\wt n}(\diff\phi\cdot\phi^{-1})\intprod \frac{\mu}{\wt n}\right)\diamond\wt n - \sum_i \frac{\delta h_i}{\delta\wt n}\diamond\wt n\circ\diff B_t^i
    \,,\\
    \hbox{where}\quad \diff\wt\phi\cdot\wt\phi^{-1} &= \frac{\delta h}{\delta\mu}\,\diff t  + \sum_i\frac{\delta h_i}{\delta\mu}\circ \diff B_t^i -\frac{a\rho}{\wt n}u\,\diff t - \sum_i\frac{a\rho}{\wt n}\xi_i\circ \diff W_t^i \label{eq:salt mu momentum eq}
    \,,\\
    \left( \diff + \mathcal{L}_{\diff\phi\cdot\phi^{-1}} \right)\rho &= 0
    \,,\\
    \left( \diff + \mathcal{L}_{\diff\phi\cdot\phi^{-1}} \right)s &= 0
    \,,\\
    \left( \diff + \mathcal{L}_{\diff\wt\phi\cdot\wt\phi^{-1}} \right)\wt n &= 0 \label{eq:salt wt n advection eq}
    \,.
\end{align}

\begin{remark}
    In the above equations, the notation $\diff\phi\cdot\phi^{-1}$ and $\diff\wt\phi\cdot\wt\phi^{-1}$ has been used to improve readability. When operations, such as Lie derivative, $\ad^*$, or insertion, are applied to these objects, it is meant that the operator acts with respect to the vector fields and the time increment ($\diff t$, $\circ\diff W_t^i$, or $\circ\diff B_t^i$) is not involved in the operation.
\end{remark}

\subsection{Hamiltonian structure of the stochastic equations}\label{subsec:Hamiltonian_stochastic}

We now perform an analogous calculation to the deterministic case, as in Section \ref{subsec:Hamiltonian_deterministic}, to arrive at the stochastic Hamiltonian structure of the equations. Namely, as before, a fluid Hamiltonian is defined as in \eqref{eqn:fluid_Legendre_transform} and we define the Hamiltonian $\mcal{H}$ as in equation \eqref{eqn:total_Hamiltonian_definition}. We consider this together with stochastic Hamiltonians, $\mcal{H}_i$, where each of these objects is defined as
\begin{align}
    \mcal{H}(M,\rho,s,\mu,\wt n) &:= h(\mu,\wt n) + \mathfrak{h}\left(M-\frac{a\rho}{\wt n}\mu , \rho, s \right)
    \,,\\
    \mcal{H}_i(M,\rho,\mu,\wt n) &:= \scp{M-\frac{a\rho}{\wt n}\mu}{\xi_i}
    \,,
\end{align}
and the Hamiltonians $h_i(\mu,\wt n)$ corresponding to the stochastic part of the magnetic dynamics written below remain the same as those in the above Hamilton-Pontryagin principle. The action then becomes
\begin{align}\label{eqn:stochastic_Hamiltonian_action}
\begin{split}
0 = \delta S & =: \delta \int \scp{\mu}{\diff\wt\phi\cdot\wt\phi^{-1}} + \scp{M}{\diff\phi\cdot\phi^{-1}} 
\\
&\qquad\qquad - \mcal{H}(M,\rho,s,\mu,\wt n)\,\diff t -\sum_i\mcal{H}_i(M,\rho,\mu,\wt n)\circ \diff W_t^i - \sum_i h_i(\mu,\wt n)\circ\diff B_t^i
\\
&\qquad\qquad +\scp{\diff\lambda_{\rho}}{\rho_0\cdot\phi^{-1} - \rho_t} + \scp{\diff\lambda_{s}}{s_0\cdot\phi^{-1} - s_t} + \scp{\diff\lambda_{\wt n}}{{\wt n}_0\cdot\wt\phi^{-1} - {\wt n}_t}\,,
\end{split}
\end{align}
Varying as in the deterministic case, with the stochastic variations as defined above, we obtain the following stochastic Hamiltonian structure
\begin{equation}\label{eqn:stochastic_LP_equations}
    \diff
    \begin{pmatrix}
        M \\ \rho \\ s \\ \mu \\ \wt n
    \end{pmatrix}
    =
    -
    \begin{pmatrix}
        \ad^*_\Box M & \Box \diamond \rho & \Box \diamond s & 0 & 0
        \\
        \mathcal{L}_\Box\rho & 0 & 0 & 0 & 0
        \\
        \mathcal{L}_\Box s & 0 & 0 & 0 & 0
        \\
        0 & 0 & 0 & \ad^*_\Box \mu & \Box \diamond \wt n
        \\
        0 & 0 & 0 & \mathcal{L}_\Box \wt n & 0
    \end{pmatrix}
    \begin{pmatrix}
        \delta\mcal{H} / \delta M \,\diff t + \delta\mcal{H}_i / \delta M \circ\diff W_t^i + \delta h_i / \delta M \circ\diff B_t^i \\ \delta\mcal{H} / \delta \rho \,\diff t + \delta\mcal{H}_i / \delta \rho \circ\diff W_t^i + \delta h_i / \delta \rho \circ\diff B_t^i \\  \delta\mcal{H} / \delta s \,\diff t + \delta\mcal{H}_i / \delta s \circ\diff W_t^i + \delta h_i / \delta s \circ\diff B_t^i \\  \delta\mcal{H} / \delta \mu \,\diff t + \delta\mcal{H}_i / \delta \mu \circ\diff W_t^i + \delta h_i / \delta \mu \circ\diff B_t^i \\  \delta\mcal{H} / \delta \wt n \,\diff t + \delta\mcal{H}_i / \delta \wt n \circ\diff W_t^i + \delta h_i / \delta \wt n \circ\diff B_t^i
    \end{pmatrix}
    \,,
\end{equation}
where, in the above equation, a sum over the index $i$ is implied. In the following example, we will demonstrate particular choices of the Hamiltonians.

\subsection{Stochastic compressible Hall MHD}\label{sec-SHMHD}

We take the Hamiltonian defined in equation \eqref{eqn:compressible_Hamiltonian} for compressible Hall MHD, and augment it with Hamiltonians $\mcal{H}_i$ and $h_i$ defined as follows
\begin{align}
    {\mcal{H}}(M, \rho, s, \mu, \wt{n}) &= \int_{\mcal{D}}\frac{1}{2\rho}\left|M - \frac{a\rho}{\wt{n}}\mu\right|^2 + \rho e(\rho, s) + \frac{R^2}{2}\left|\mb{d}\frac{\mu}{\wt{n}}\right|^2 dV
    \label{eqn:stochastic_compressible_Hamiltonians1}
    \,,\\
    \mcal{H}_i(M, \mu, \rho, \wt{n}) &= \scp{M - \frac{a\rho}{\wt{n}}\mu}{\xi_i}_{\mathfrak{X}^*(\mcal{D})\times \mathfrak{X}(\mcal{D})} 
    \label{eqn:stochastic_compressible_Hamiltonians2}
    \,,\\
    h_i(\mu,\wt n) &= \int_{\mcal{D}} (\bs{\delta}\sigma_i)^\sharp \intprod \frac{\mu}{\wt{n}} \,dV 
    \label{eqn:stochastic_compressible_Hamiltonians3}
    \,,
\end{align}
where $\xi\in\mathfrak{X}(\mathcal{D})$ and $\sigma_i \in \Lambda^2(\mathcal{D})$ are the exogenously prescribed noise vector fields and $2$-forms respectively which are chosen to calibrate the system to data or observations. Furthermore, note that by $|\mu|^2$ for $\mu = \bs{\mu}\cdot d\bx\otimes dV\in\mathfrak{X}^*$, for example, we mean $\mu^{\sharp} \intprod (\bs{\mu}\cdot d\bx$).
\begin{remark}[A three dimensional Euclidean domain] When the domain is taken to be a three dimensional Euclidean space, these Hamiltonians become
\begin{align}
    {\mcal{H}}(M, \rho, s, \mu, \wt{n}) &= \int_{\mcal{D}}\frac{1}{2\rho}\left|\bs{M} - \frac{a\rho}{\wt{n}}\bs{\mu}\right|^2 + \rho e(\rho, s) + \frac{R^2}{2}\left|\operatorname{curl}\frac{\bs{\mu}}{\wt{n}}\right|^2 d^3x
    \label{eqn:stochastic_compressible_Hamiltonians1_vector_calculus}
    \,,\\
    \mcal{H}_i(M, \mu, \rho, \wt{n}) &= \int_{\mcal{D}} \left(\bs{M} - \frac{a\rho}{\wt{n}}\bs{\mu}\right)\cdot\bs{\xi}_i\,d^3x
    \label{eqn:stochastic_compressible_Hamiltonians2_vector_calculus}
    \,,\\
    h_i(\mu,\wt n) &= \int_{\mcal{D}} \frac{\bs{\mu}}{\wt{n}}\cdot \operatorname{curl} \bs{\sigma}_i   \,d^3x
    \label{eqn:stochastic_compressible_Hamiltonians3_vector_calculus}
    \,,
\end{align}
and both $\bs{\xi}_i$ and $\bs{\sigma}_i$ can be understood as vectors.
\end{remark}
\begin{remark}[On the choice of the Hamiltonians for the noise]
    The Hamiltonian $\mcal{H}_i$ corresponds to the inclusion of \emph{transport noise} \cite{Holm2015} into the fluid part of the dynamics. To see this, notice that $\mcal{H}_i$, as defined above, is the stochastic Hamiltonian corresponding to transport-type noise \cite{ST2023}, namely $\mcal{H}_i = \scp{\pi}{\xi_i}$ where $\pi$ is the fluid momentum. The Hamiltonian $h_i$ has been chosen such that its variational derivative in $\mu$ has a similar form to the same derivative of the deterministic Hamiltonian corresponding to the magnetic part of the dynamics. Namely, that the operator ${\rm curl}(\cdot) / \wt n$ appears in the expression. As we will see, this ensures that charge neutrality is preserved after the addition of noise to the system.

    Some freedom of choice of the stochastic Hamiltonian $h_i(\mu, \wt{n})$ for the electron degree of freedom still remains, in addition to the stochastic spatial translations of the ion Lagrangian fluid parcels. These other choices of Hamiltonians $h_i(\mu, \wt{n})$ introduce a variety of opportunities for influencing the stochastic dynamics of the electronic degree of freedom. For example, directing microwave radiation into the Hall MHD flow would tend to primarily influence the lighter electrons at first, and then the excited electron degree of freedom would influence the ions. Further investigation of the physics underlying this opportuny may provide a means of selectively controlling Hall MHD flow.
\end{remark}

The variations of the Hamiltonian, $\mcal{H}$, corresponding to the drift part of the equations are calculated as follows
\begin{align*}
    \frac{\delta \mcal{H}}{\delta M} &= \left[\frac{1}{\rho}\left(M - \frac{a\rho}{\wt n}\mu \right)\right]^\sharp = u \,,\qquad \frac{\delta\mcal{H}}{\delta s} = \rho T
    \,,\\
    \frac{\delta\mcal{H}}{\delta \mu} &= -\left[\frac{a}{\wt n}\left(M - \frac{a\rho}{\wt n}\mu \right)\right]^\sharp + \frac{R^2}{\star\wt{n}}\left(\bs{\delta} \mb{d}\frac{\mu}{\wt{n}}\right)^\sharp := - \frac{a\rho}{\wt n}u + v_H =: v
    \,,\\
    \frac{\delta\mcal{H}}{\delta\rho} &= \wh{h}(p,s) - \frac{1}{2\rho^2}\left|\bs{M} - \frac{a\rho}{\wt{n}}\bs{\mu}\right|^2 - \left[\frac{1}{\rho}\left({M} - \frac{a\rho}{\wt{n}}{\mu}\right)\right]^\sharp\intprod\frac{a\mu}{\wt n} = \wh{h}(p,s) - \frac{|u|^2}{2} - a\,\left(u\intprod\frac{\mu}{\wt n}\right)
    \,\\
    \frac{\delta\mcal{H}}{\delta\wt n} &= \left( \frac{1}{\rho}\left( M - \frac{a\rho}{\wt n}\mu  \right) \right)^\sharp \intprod \left( \frac{a\rho}{\wt n^2}\mu \right) - \frac{R^2}{\star\wt{n}}\left(\bs{\delta} \mb{d}\frac{\mu}{\wt{n}}\right)^\sharp \intprod \frac{\mu}{\wt{n}} = \frac{a\rho}{\wt n}\,\left(u\intprod\frac{\mu}{\wt n}\right) + \frac{\delta h}{\delta\wt n}
    \,,
\end{align*}
which are the coordinate free expressions corresponding to the variations in equation \eqref{eqn:Hamiltonian_variations}, computed when illustrating the deterministic case. Varying the stochastic Hamiltonians corresponding to the `diffusion' part of the equations, we have
\begin{align*}
    \frac{\delta\mcal{H}_i}{\delta M} &= \xi_i \,,\quad \frac{\delta\mcal{H}_i}{\delta \mu} = -\frac{a\rho}{\wt n}\xi_i \,,\quad \frac{\delta\mcal{H}_i}{\delta \rho} = -a
    \,\xi_i \intprod \frac{\mu}{\wt n}
    \,,\\
    \frac{\delta\mcal{H}_i}{\delta \wt n} &= \frac{a\rho}{\wt n}\xi_i\intprod\frac{\mu}{\wt n} \,,\quad \frac{\delta h_i}{\delta \mu} = \left( \frac{\bs{\delta}\sigma_i\,dV}{\wt n} \right)^\sharp \,,\quad \frac{\delta h_i}{\delta \wt n} = -  \left(\frac{\bs{\delta}\sigma_i\,dV}{\wt n}\right)^\sharp \intprod \frac{\mu}{\wt{n}}  \,.
\end{align*}

We begin by noticing that varying the Hamiltonians with respect to $M$ and $\mu$ respectively gives
\begin{align}
    \rmd\phi\cdot\phi^{-1} &= \frac{\delta\mcal{H}}{\delta M}\,\diff t + \sum_i \frac{\delta\mcal{H}_i}{\delta M}\circ\diff W_t^i = u\,\diff t + \sum_i \xi_i\circ\diff W_t^i
    \,,\\
    \rmd\wt\phi\cdot\wt\phi^{-1} &= \frac{\delta\mcal{H}}{\delta \mu}\,\diff t + \sum_i \frac{\delta\mcal{H}_i}{\delta \mu}\circ\diff W_t^i + \sum_i\frac{\delta h_i}{\delta\mu}\circ\diff B_t^i 
    \\
    &= \left(-\frac{a\rho}{\wt n}u + \frac{R^2}{\star\wt n}\left( \bs{\delta}\mb{d}\frac{\mu}{\wt n}\right)^\sharp \right)\,\diff t - \frac{a\rho}{\wt n}\sum_i\xi_i\circ \diff W_t^i + \sum_i\left( \frac{\bs{\delta}\sigma_i\,dV}{\wt n} \right)^\sharp \circ\diff B_t^i 
    \,,\\
    &= -\frac{a\rho}{\wt n}\,\diff\phi\cdot\phi^{-1} + v_H\,\diff t + \sum_i\left( \frac{\bs{\delta}\sigma_i\,dV}{\wt n} \right)^\sharp \circ\diff B_t^i \,.
\end{align}
\begin{proposition}
    When the law of charge neutrality holds for the initial conditions $\rho_0$ and $\wt{n}_0$, then it is preserved under the stochastic flow of the system. That is:
    \begin{equation}
        a\rho_0 + \wt{n}_0 = 0 \quad\Longrightarrow\quad a\rho + \wt n = 0 \,,
    \end{equation}
    for all time $t$.
\end{proposition}
\begin{proof}
    From the stochastic Lie-Poisson equations \eqref{eqn:stochastic_LP_equations} and the definition of the Hamiltonians \eqref{eqn:stochastic_compressible_Hamiltonians1}-\eqref{eqn:stochastic_compressible_Hamiltonians3}, we may deduce the equations for $\rho$ and $\wt n$ to be
    \begin{align}
        \diff \rho + \mathcal{L}_u \rho \diff t + \sum_i \mathcal{L}_{\xi_i}\rho \circ \diff W_t^i &= 0
        \,,\\
        \diff \wt n + \mathcal{L}_{\frac{a\rho}{\wt n}u}\wt n\,\diff t + \sum_i \mathcal{L}_{\frac{a\rho}{\wt n}\xi_i}\wt n\circ \diff W_t^i + \sum_i \mathcal{L}_{\left( \frac{\bs{\delta}\sigma_i dV}{\wt n} \right)^\sharp}\wt n \circ \diff B_t^i &= 0 \,.
    \end{align}
    Notice that, for a density $D$ and a vector field $X$, we have\footnote{In three dimensional vector calculus notation on Euclidean domains, maintaining the notation $D = \star D\,dV$, this expression translates to $\mathcal{L}_X(\star D\,dV) = \nabla\cdot(\star D\bs{X}) \,dV$.} $\mathcal{L}_X D = \mb{d}(\star {(\star D)} X^\flat)$, where we have employed the identity $D = \star D\,dV$ discussed in Remark \ref{remark:dividing_by_density}. Also exploiting the fact that $a$ is constant, the equation for $\wt n$ is
    \begin{align*}
        \diff\wt n &= a\mb{d}(\star
        {(\star\rho)} u^\flat)\,\diff t + a\sum_i \mb{d}(\star{(\star\rho)} \xi_i^\flat)\circ\diff W_t^i + \sum_i \mb{d}(\ast\bs{\delta}\sigma_i)\circ\diff B_t^i
        \\
        &= a\mathcal{L}_{u}\rho\,\diff t + a\sum_i\mathcal{L}_{\xi_i}\rho\circ\diff W_t^i \,,
    \end{align*}
    the second line follows from the first since $\mb{d}\star\bs{\delta}(\bs{A}\cdot d\bx) = 0$ for all $1$-forms $\bs{A}\cdot d\bx$, which is analogous to the more familiar identity $\operatorname{div}\operatorname{curl}\bs{A} = 0$, for all $3$-vectors $\bs{A}$.
    Combining these equations, we have
    \begin{equation*}
        \diff\,(a\rho + \wt n) = 0 \,,
    \end{equation*}
    which concludes the proof that charge neutrality is preserved by Hall MHD, provided it holds initially.
\end{proof}
 
In the above proof, we have given the equations for $\rho$ and $\wt n$. The remaining equations are as follows.

\begin{align}
    \begin{split}
    \left(\diff + \ad^*_{\diff\phi\cdot\phi^{-1}}\right)M &= -\rho\,\mb{d}\left(\wh{h}(p,s) - \frac{|u|^2}{2} - a\,\left(u\intprod\frac{\mu}{\wt n}\right) \right)\,\diff t + \sum_i \rho\,\mb{d}\left( a\xi_i\intprod\frac{\mu}{\wt n} \right)\circ \diff W_t^i + \rho T\,\mb{d}s\diff t
    \\
    &= -\mb{d}p\,\diff t + \rho\,\mb{d}\left(\frac{|u|^2}{2} + a\,\left(u\intprod\frac{\mu}{\wt n}\right) \right)\,\diff t + \sum_i \rho\,\mb{d}\left( a\xi_i\intprod\frac{\mu}{\wt n} \right)\circ \diff W_t^i
    \,,
    \end{split}
    \\
    \begin{split}
    \left(\diff  + \ad^*_{\diff\wt\phi\cdot\wt\phi^{-1}}\right)\frac{\mu}{\wt n} &= -\mb{d}\left(\frac{a\rho}{\wt n}\,\left(u\intprod\frac{\mu}{\wt n}\right) - \frac{R^2}{\star\wt{n}}\left(\bs{\delta} \mb{d}\frac{\mu}{\wt{n}}\right)^\sharp \intprod \frac{\mu}{\wt{n}} \right)\diff t 
    \\
    &\qquad - \sum_i\mb{d}\left( \frac{a\rho}{\wt n}\xi_i\intprod \frac{\mu}{\wt n} \right)\circ\diff W_t^i +\sum_i \mb{d}\left( \left(\frac{\bs{\delta}\sigma_i\,dV}{\wt n}\right)^\sharp \intprod \frac{\mu}{\wt{n}} \right)\circ\diff B_t^i
    \,,
    \end{split}
    \\
    \left( \diff + \mathcal{L}_{\diff\phi\cdot\phi^{-1}} \right)s &= 0
    \,.
\end{align}
Since $a\rho / \wt n = -1$, we may evaluate the second of these equations as
\begin{equation}
\begin{aligned}
    \left(\diff + \ad^*_{\diff\phi\cdot\phi^{-1}}\right)\frac{\mu}{\wt n} &= - \mathcal{L}_{\frac{R^2}{\wt n}\left( \bs{\delta}\mb{d}\frac{\mu}{\wt n}\right)^\sharp}\frac{\mu}{\wt n}\,\diff t - \sum_i \mathcal{L}_{\left( \frac{\bs{\delta}\sigma_i\,dV}{\wt n} \right)^\sharp} \frac{\mu}{\wt n}\circ\diff B_t^i +\mb{d}\left( u\intprod\frac{\mu}{\wt n} + \frac{R^2}{\star\wt{n}}\left(\bs{\delta} \mb{d}\frac{\mu}{\wt{n}}\right)^\sharp \intprod \frac{\mu}{\wt{n}} \right)\diff t
    \\
    &\qquad + \mb{d}\left( \xi_i\intprod \frac{\mu}{\wt n} \right)\circ\diff W_t^i + \mb{d}\left( \left(\frac{\bs{\delta}\sigma_i\,dV}{\wt n}\right)^\sharp \intprod \frac{\mu}{\wt{n}} \right)\circ\diff B_t^i
    \,,\\
    &= \mb{d} \left( \left( u\,\diff t + \sum_i\xi_i\circ\diff W_t^i \right)\intprod\frac{\mu}{\wt n} \right) - \frac{R^2}{\star\wt{n}}\left(\bs{\delta} \mb{d}\frac{\mu}{\wt{n}}\right)^\sharp \intprod \mb{d}\frac{\mu}{\wt n}\,\diff t 
    \\
    &\qquad - \sum_i\left(\frac{\bs{\delta}\sigma_i\,dV}{\wt n}\right)^\sharp \intprod \mb{d}\frac{\mu}{\wt n}\circ\diff B_t^i \,.
\end{aligned}
\end{equation}
Noting that $M = \rho u^\flat + a\mu/\wt n$, we have the following equation for the fluid velocity $u^\flat$
\begin{equation}
    (\diff + \ad^*_{\diff\phi\cdot\phi^{-1}})u^\flat = -\frac{1}{\rho}\mb{d}p\,\diff t + \mb{d}\frac{|u|^2}{2} + \frac{aR^2}{\star\wt{n}}\left(\bs{\delta} \mb{d}\frac{\mu}{\wt{n}}\right)^\sharp \intprod \mb{d}\frac{\mu}{\wt n}\,\diff t + a\sum_i\left(\frac{\bs{\delta}\sigma_i\,dV}{\wt n}\right)^\sharp \intprod \mb{d}\frac{\mu}{\wt n}\circ\diff B_t^i \,.
\end{equation}

\begin{remark}[The possibility of stochastic pressure terms.]
    When formulating incompressible stochastic fluid models, the inclusion of stochastic integration necessitates semimartingale pressure terms \cite{SC2021}. Indeed, if formulating the incompressible example from Section \ref{subsec:incompressible_example} we would see the need for such terms. In the stochastic compressible equations considered in this section, such pressure terms are not mathematically required. Nonetheless, it is possible to include a full semimartingale pressure in the equation of motion by including terms of the form $\rho e_i(\rho,s)$ into the Hamiltonians $\mcal{H}_i$ and $h_i$, where $e_i$ satisfy equations of state $\mb{d}e_i=-p_i\mb{d}\rho^{-1} + T_i\mb{d}s$ for some $p_i$ and $T_i$.
\end{remark}

\subsection{Energy preserving stochastic perturbations} \label{sec: SFLT}
When introducing stochasticity into the Hamilton--Pontryagin variational principle for the Hall MHD equations, there are no unique choices of stochastic perturbations. As demonstrated in equation \eqref{eqn:minimally_coupled_action_Pontryagin_stochastic}, one particular type of stochasticity takes the form of stochastic Hamiltonians $\scp{\pi}{\xi_i}$ and $h_i(\mu, \wt{n})$. These Hamiltonians generate stochastic transport vector fields and the stochastic perturbations are considered to be the transport noise. 

Another class of stochastic perturbations derived in \cite{HH2021} and known as Stochastic Forcing by Lie Transport (SFLT) introduces stochastic \emph{forces} rather than stochastic \emph{transport}. These stochastic forces are imposed via a stochastic reduced Lagrange d'Alembert variational principle and they designed to preserve the domain integrated energy of the unperturbed dynamics. 

To construct the SFLT stochasticity, we again consider the set of driving semimartingales $\{W^i_t\}_{i=1,\ldots}$ and $\{B^i_t\}_{i=1,\ldots}$ which are two collections of i.i.d Brownian motions. Here, $W^i_t$ are the driving stochastic processes for the fluidic part of the dynamics and $B^i_t$ are the driving stochastic processes for the magnetic part of the dynamics. We define the amplitudes of the stochastic perturbations using quantities taking values in the Lie co-algebra \eqref{eqn:Hall MHD Lie co algebra}. Namely, we have the set of amplitudes $\{f^M_i\}_{i=1,\ldots} \in \mathfrak{X}^*(\mcal{{D}})$, $\{f^\rho_i\}_{i=1,\ldots} \in \Lambda^n(\mcal{{D}})$, $\{f^s_i\}_{i=1,\ldots} \in \Lambda^0(\mcal{{D}})$, $\{f^\mu_i\}_{i=1,\ldots} \in \mathfrak{X}^*(\mcal{{D}})$ and $\{f^{\wt{n}}_i\}_{i=1,\ldots} \in \Lambda^n(\mcal{{D}})$. These quantities define the stochastic perturbations to $M$, $\rho$, $s$, $\mu$ and $\wt{n}$ respectively. 

The variational principle we use to derive the SFLT ideal compressible Hall MHD equations is the stochastic reduced Lagrange-d'Alembert variational principle. As the Hall MHD equations are Lie-Poisson equation on the Lie co-algebra \eqref{eqn:Hall MHD Lie co algebra}, we consider the stochastic reduced Lagrange-d'Alembert principle defined on 
\begin{align*}
    \mathfrak{X}(\mcal{D}) \ltimes \left(\Lambda^0(\mcal{D})\oplus \Lambda^n(\mcal{D})\right)\oplus \left(\mathfrak{X}(\mcal{D})\ltimes \Lambda^0(\mcal{D})\right)\,,
\end{align*}
and its dual algebra \eqref{eqn:Hall MHD Lie co algebra}. Using the Hamiltonian defined in equation \eqref{eqn:total_Hamiltonian_definition}, this variational principle can be written as 
\begin{align}\label{eqn:minimally_coupled_action_stochastic_sflt}
\begin{split}
0 = \delta S & =: \delta \int \scp{\mu}{\diff\wt\phi\cdot\wt\phi^{-1}} + \scp{\wt{n}}{\diff v_{\wt{n}}\cdot\wt\phi^{-1}} + \scp{M}{\diff\phi\cdot\phi^{-1}} + \scp{\rho}{\diff v_\rho \cdot\phi^{-1}} + \scp{s}{\diff v_s\cdot\phi^{-1}}\\
&\qquad \qquad - \mcal{H}(M,\rho,s,\mu,\wt n)\,\diff t \\
& \qquad + \int \scp{\ad^*_{(\frac{\delta \mcal{H}}{\delta M},\frac{\delta \mcal{H}}{\delta \rho},\frac{\delta \mcal{H}}{\delta s})}(f^M_i, f^\rho_i, f^s_i)}{(\eta, \eta_\rho, \eta_s)}\circ \diff W^i_t + \scp{\ad^*_{(\frac{\delta \mcal{H}}{\delta \mu},\frac{\delta \mcal{H}}{\delta \wt{n}})}(f^\mu_i, f^{\wt{n}}_i)}{(\wt{\eta}, \eta_{\wt{n}})}\circ \diff B^i_t
\end{split}
\end{align}
where the summation over the index $i$ between the force amplitudes $f^{(\cdot)}_i$ and the stochastic processes are assumed. 
In the principle \eqref{eqn:minimally_coupled_action_stochastic_sflt}, we make use of additional group valued quantities $v_{\rho}$, $v_s$ and $v_{\wt{n}}$ which are elements of the semi-direct product groups
\begin{align*}
    (\phi, v_\rho, v_s) \in \operatorname{Diff}_1(\mcal{D})\ltimes(\Lambda^0(\mcal{D})\times \Lambda^n(\mcal{D})) = \mathfrak{s}_1\,,\quad (\wt{\phi}, v_{\wt{n}}) \in \operatorname{Diff}_2(\mcal{D})\ltimes\Lambda^0(\mcal{D}) = \mathfrak{s}_2\,.
\end{align*}
Extending the construction in \cite{ST2023} for the variations of $\diff \phi\cdot\phi^{-1} \in \mathfrak{X}(\mcal{D})$ to the semi-direct product case, we have the analogous variations of $\diff v_\rho \cdot \phi^{-1}$, $\diff v_s \cdot \phi^{-1}$ and $\diff v_{\wt{n}} \cdot \wt{\phi}^{-1}$. In particular, for an arbitrarily chosen time-differentiable $(\eta,\eta_\rho,\eta_s) = (\delta\phi,\delta v_\rho, \delta v_s)\cdot\phi^{-1} = (\delta \phi \cdot \phi^{-1},\delta v_\rho \cdot \phi^{-1},\delta v_s \cdot \phi^{-1})$, and $(\wt\eta,\eta_{\wt n}) = (\delta\wt\phi,\delta v_{\wt n})\cdot\wt \phi^{-1} = (\delta \wt{\phi} \cdot \wt{\phi}^{-1},\delta v_{\wt{n}} \cdot \wt{\phi}^{-1})$ we have
\begin{align}
\begin{split}
    \delta(\diff \phi \cdot \phi^{-1}, \diff v_\rho \cdot \phi^{-1}, \diff v_s \cdot \phi^{-1}) 
    &= \p_t \, (\eta, \eta_\rho, \eta_s)\,\diff t - \ad_{(\diff \phi \cdot \phi^{-1}, \diff v_\rho \cdot \phi^{-1}, \diff v_s \cdot \phi^{-1})}(\eta, \eta_\rho, \eta_s)\\ 
    &\hspace{-130pt}= (\p_t \eta\,\diff t - \ad_{\diff\phi\cdot\phi^{-1}}\eta, \p_t \eta_\rho\,\diff t - \mcal{L}_{\diff \phi\cdot \phi^{-1}}\eta_\rho - \mcal{L}_{\eta}\diff v_\rho \cdot \phi^{-1}, \p_t \eta_s\,\diff t - \mcal{L}_{\diff \phi\cdot \phi^{-1}}\eta_s - \mcal{L}_{\eta}\diff v_s \cdot \phi^{-1})\,,
\end{split}
\\
\begin{split}
    \delta(\diff \widetilde\phi \cdot \widetilde\phi^{-1}, v_{\wt{n}}\cdot \wt{\phi}^{-1}) &=\p_t\, (\wt{\eta}, \eta_{\wt{n}})\,\diff t - \ad_{(\diff \wt{\phi} \cdot \wt{\phi}^{-1}, \diff v_{\wt{n}} \cdot \wt{\phi}^{-1})}(\wt{\eta}, \eta_{\wt{n}})\\
    &= (\p_t \eta\,\diff t - \ad_{\diff \wt{\phi}\cdot\wt{\phi}^{-1}} \eta, \p_t \eta_{\wt{n}}\,\diff t - \mcal{L}_{\diff \wt{\phi}\cdot \wt{\phi}^{-1}}\eta_{\wt{n}} - \mcal{L}_{\wt{\eta}}\diff v_{\wt{n}} \cdot \wt{\phi}^{-1})
    \,.
\end{split}
\end{align}
As in Section \ref{subsec:stochastic_variational_principle} and following \cite{ST2023}, the Lie algebra element $(\eta,\eta_\rho,\eta_s,\wt\eta,\eta_{\wt n})$ is chosen arbitrarily and the variations in the group are constructed in terms of this. These constrained variations form the extension to the variations of $\diff \phi\cdot\phi^{-1}$ and $\diff \wt{\phi}\cdot \wt{\phi}^{-1}$ given in equation \eqref{eqn:stochastic_Lin_constraints}. Taking the variations and applying the stochastic fundamental theorem of calculus we obtain the following system of SPDES expressed in their geometric form
\begin{align}
\begin{split}
    \left( \diff + \ad^*_{\diff\phi\cdot\phi^{-1}} \right)M &= -\frac{\delta \mcal{H}}{\delta \rho}\diamond\rho \,\diff t - \frac{\delta \mathcal{H}}{\delta s}\diamond s\,\diff t - \ad^*_u f^M_i \circ \diff W^i_t - \frac{\delta \mcal{H}}{\delta \rho}\diamond f^\rho_i \,\diff W^i_t -\frac{\delta \mcal{H}}{\delta s}\diamond f^s_i \circ \diff W^i_t\,, \\
    &\text{where}\quad \diff\phi\cdot\phi^{-1} = \frac{\delta \mcal{H}}{\delta M}\,\diff t = u\,\diff t \,,\\
    \left( \diff + \ad^*_{\diff\wt\phi\cdot\wt\phi^{-1}} \right)\mu &= -\frac{\delta \mcal{H}}{\delta \wt{n}}\diamond \wt{n}\,\diff t - \ad^*_v f^\mu_i \circ \diff B^i_t - \frac{\delta \mcal{H}}{\delta \wt{n}}\diamond f^{\wt{n}}_i \circ \diff B^i_t
    \,,\\
    &\text{where}\quad \diff\wt\phi\cdot\wt\phi^{-1} = \frac{\delta \mcal{H}}{\delta\mu}\,\diff t = v\,\diff t
    \,,\\
    \left( \diff + \mathcal{L}_{\diff\phi\cdot\phi^{-1}} \right)\rho &= -\mcal{L}_u f^\rho_i \circ \diff W^i_t
    \,,\\
    \left( \diff + \mathcal{L}_{\diff\phi\cdot\phi^{-1}} \right)s &= -\mcal{L}_u f^s_i \circ \diff W^i_t
    \,,\\
    \left( \diff + \mathcal{L}_{\diff\wt\phi\cdot\wt\phi^{-1}} \right)\wt n &= -\mcal{L}_v f^{\wt{n}}_i \circ \diff B^i_t 
    \,.
\end{split}\label{eq:sflt general lp eq}
\end{align}
As indicated in \cite{HH2021}, the SFLT stochastic perturbations can be organised into a Hamiltonian structure. In the case of the Lie-Poisson equation \eqref{eq:sflt general lp eq}, we have the following 
\begin{align}
    \begin{split}
        \diff
    \begin{pmatrix}
        M \\ \rho \\ s \\ \mu \\ \wt n
    \end{pmatrix}
    &=
    -
    \begin{pmatrix}
        \ad^*_\Box M & \Box \diamond \rho & \Box \diamond s & 0 & 0
        \\
        \mathcal{L}_\Box\rho & 0 & 0 & 0 & 0
        \\
        \mathcal{L}_\Box s & 0 & 0 & 0 & 0
        \\
        0 & 0 & 0 & \ad^*_\Box \mu & \Box \diamond \wt n
        \\
        0 & 0 & 0 & \mathcal{L}_\Box \wt n & 0
    \end{pmatrix}
    \begin{pmatrix}
            \delta \mcal{H} / \delta M \\
            \delta \mcal{H} / \delta \rho \\
            \delta \mcal{H} / \delta s \\
            \delta \mcal{H} / \delta \mu \\
            \delta \mcal{H} / \delta \wt{n} 
    \end{pmatrix}
    \diff t
    \\
    &-
    \begin{pmatrix}
        \ad^*_\Box f^M_i\circ \diff W^i_t & \Box \diamond f^\rho_i\circ \diff W^i_t & \Box \diamond f^s_i\circ \diff W^i_t & 0 & 0
        \\
        \mathcal{L}_\Box f^\rho_i\circ \diff W^i_t & 0 & 0 & 0 & 0
        \\
        \mathcal{L}_\Box f^s_i\circ \diff W^i_t & 0 & 0 & 0 & 0
        \\
        0 & 0 & 0 & \ad^*_\Box f^\mu_i\circ \diff B^i_t & \Box \diamond f^{\wt{n}}_i\circ \diff B^i_t
        \\
        0 & 0 & 0 & \mathcal{L}_\Box f^{\wt{n}}_i\circ \diff B^i_t & 0
    \end{pmatrix}
    \begin{pmatrix}
            \delta \mcal{H} / \delta M \\
            \delta \mcal{H} / \delta \rho \\
            \delta \mcal{H} / \delta s \\
            \delta \mcal{H} / \delta \mu \\
            \delta \mcal{H} / \delta \wt{n} 
    \end{pmatrix}
    \,,
    \end{split}
\end{align}
Compared to the deterministic Hamiltonian structure in equation \eqref{eqn:Hall MHD PB}, the amplitudes of the SFLT stochastic perturbations form a frozen Lie-Poisson bracket on the Lie co-algebra \eqref{eqn:Hall MHD Lie co algebra} which is then contracted with Brownian processes.

For applications to fluid flows, the volume density perturbations $f^\rho_i$ are typically set to zero to preserve the advection property of volume density due to its relation to the determinant of the fluid back-to-labels map. For the preservation of the charge neutrality condition \eqref{chargeNeutral2}, we additionally set the charge density perturbations $f^{\wt{n}}_i$ to vanish. Under these conditions, by inserting the Hamiltonian for the ideal compressible Hall MHD equations given in \eqref{eqn:total_Hamiltonian_definition}, we obtain the stochastically perturbed Ohm's law for the dynamics of the magnetic potential $A$
\begin{align}
    \diff A + \mcal{L}_v A\,\diff t + \frac{R}{\wt{n}}\mcal{L}_v f^\mu_i \circ \diff B^i_t = \mb{d}(v\intprod A)\,\diff t\,.
\end{align}
Following similar calculations as in the deterministic case, the full SFLT Hall MHD equations can be written exclusively in the variables $u, \rho, s, A$ and $B$ to have
\begin{align}
    \begin{split}
    &\diff u^\flat + \mcal{L}_u u^\flat = \frac{1}{2}\mb{d}\left(u\intprod u^\flat\right)\,\diff t - \frac{1}{\star\rho}\left(\mb{d}p(e,s) + (\bs{\delta}\mb{d}A)^\sharp \intprod \mb{d}A\right)\,\diff t\\
    &\qquad \qquad \qquad \qquad \qquad - \frac{1}{\rho}\mcal{L}_u f^M_i\circ \diff W^i_t - \frac{1}{\rho}\mcal{L}_v f^\mu_i\circ \diff B^i_t + T\mb{d}f^s_i \circ \diff W^i_t\,,\\
    & \diff A + \mcal{L}_v A\,\diff t - \frac{R}{a\rho}\mcal{L}_v f^\mu_i \circ \diff B^i_t = \mb{d}(v\intprod A)\,\diff t \,,\\
    &\diff s + \mathcal{L}_u s\,\diff t + \mcal{L}_u f^s_i \circ \diff W^i_t = 0 \,, \\
    &\diff\rho + \mathcal{L}_u \rho\,\diff t  = 0
    \,, \quad \text{with} \quad v = u - \frac{R}{a(\star\rho)}(\bs{\delta}\mb{d}A)^\sharp\,.
\end{split}
\end{align}
We intend to bring the SFLT approach forward together with the SALT approach for the purpose of developing new stochastic forces for Hall MHD. The opportunities introduced by the SFLT approach deserve more investigation.

\section{Open problems and plans for future work}\label{sec-remarks}

In this paper the coordinate-free Eulerian equations of Hall MHD have been derived by using two different  methods SALT and SFLt that each preserve the relabelling symmetry of fluid dynamics induced in the passage from the Lagrangian, to the Eulerian representation. Taking this step enabled the SALT approach to be applied to introduce a type of stochastic transport that preserves the coadjoint Lie structure of the deterministic Hall MHD equations proved in the first step. Taking this step enabled the SFLT approach to deal with stochasti processes in Hall MHD which would preserve its total energy. 

By taking Lagrangian expectations of SALT, one creates LA SALT \cite{DHL2020,D+H2020} . The expected solutions of LA SALT satisfy deterministic equations, while the fluctuations satisfy linear equations. As a result, the LA SALT approach provides deterministic equations for their variances and higher moments. These deterministic equations for the statistics of the expected solutions in turn provide deterministic predictions of how the `climate' is changing for the stochastic SALT equations, in Ed Lorenz's sense that: ``Climate is what you expect and Weather is what you get."

In other future research, we plan to use LA SALT (Lagrangian-Averaged Stochastic Advection by Lie Transport) \cite{D+H2020, DHL2020} to derive the expected equations and the dynamical equations for the variances of the fluctuations of the solutions away from its expected solution. After establishing the LA SALT equations for Hall MHD we plan to introduce double-bracket dissipation in order to characterise the invariant measures induced by this form of dissipation. The LA SALT approach preserves the coadjoint Lie structure of the deterministic, the stochastic (SALT) equations. Moreover, the combination of the expectation/fluctuation decomposition of LA SALT and the application of double bracket dissipation preserves the coadjoint structure of the approach to an invariant measure \cite{ACH2018,DHP2023}.  

Both Lagrangian averaging (LA) in SALT and Eulerian averaging (EA) in SFLT are possible. These complementary  stochastic approaches -- Langrangian transport and Eulerian forcing -- may lead to a variety of opportunities to develop complementary \emph{expectation dynamics} for fluids and and fluid plasmas.  

We also plan to address the open problem of deriving modifications such as Levy area when introducing non-Brownian stochasticity such as coloured noise, fractional diffusion and/or geometric rough paths, following the pioneering work of \cite{CrisanHolmLeahyNilssen2022}. This more general setting for transport noise may potentially reveal more details about the invariant measures for this class of stochastic fluid dynamics, \cite{ACH2018,DHP2023}. 

In particular, the following outstanding question and its counterpart remain: 
\begin{enumerate}
    \item 
After the SALT Kelvin-Noether theorem for Brownian transport noise has been derived for a certain fluid model, does the addition of selective decay (e.g., double bracket dissipation) guarantee Gibbs invariant measure? If so, the solution behaviour would be determined and expectations could be taken in the context of standard statistical physics. This fundamental mathematical question is likely be the most difficult among the questions we wish to answer.
    \item
If Gibbs invariant measure is not guaranteed, then under what conditions for Brownian transport noise with SALT in an Euler-Poincare fluid model does the addition of selective decay (or double bracket dissipation) lead to Gibbs invariant measure?
\end{enumerate}

The current work could also allow consideration of the following applications:

$\bullet\quad$ Stochastic variational principles with additional forces for control in fluid dynamics
\cite{Troutman2012}.

$\bullet\quad$ Other effects of dissipation, e.g., by selective decay (double-bracket dissipation)
\cite{FGB_H2014}.

$\bullet\quad$ Uncertainty quantification in space weather plasmas, \cite{Gombosi_SW_2021}.

$\bullet\quad$ Other space plasmas, including solar flares, solar winds and the Earth's magnetic field  
\cite{Huba2023}.

All of these opportunities await exploration via the complementary geometric mechanics approaches of stochastic transport (SALT) and stochastic forcing (SFLT). 

\subsection*{Acknowledgements} 
We are grateful to C. Cotter, D. Crisan, H. Dumpty and M. Yan for several thoughtful suggestions during the course of this work which have improved or clarified the interpretation of its results. DH and RH were partially supported during the present work by Office of Naval Research (ONR) grant award N00014-22-1-2082, Stochastic Parameterization of Ocean Turbulence for Observational Networks. DH and OS were partially supported during the present work by European Research Council (ERC) Synergy grant Stochastic Transport in Upper Ocean Dynamics (STUOD) -- DLV-856408.


\begin{thebibliography}{99}


\bibitem{ACC2014}
Arnaudon, M., Chen, X., and Cruzeiro, A. B., 2014. Stochastic Euler-Poincar\'e reduction. J. Math. Phys. 1 August 2014; 55 (8): 081507. \url{https://doi.org/10.1063/1.4893357}

\bibitem{ACH2018}
Arnaudon, A., L de Castro, A., Holm, D.D., 2018.
Noise and Dissipation on Coadjoint Orbits.
J. Nonlin. Sci.  28:91--145.     
\url{https://doi.org/10.1007/s00332-017-9404-3}

\bibitem{AK2021}
Arnold, V.I. and Khesin, B.A., 2021. Topological Methods in Hydrodynamics. \url{https://doi.org/10.1007/978-3-030-74278-2}

\bibitem{Besse2023}
Besse, N., 2023. Stochastic Lagrangian perturbation of Lie transport and applications to fluids. Nonlinear Analysis, 232, p.113249. \url{https://doi.org/10.1016/j.na.2023.113249}

\bibitem{Brag1965}
Braginsky, S.I., 1965. in: M.A. Leontovich (ed.), Reviews of Plasma Physics, 1, Consultants Bureau,
New York, p. 205.

\bibitem{BragRoberts1995}
Braginsky, S.I. and Roberts, P.H., 1995.
Equations governing convection in earth's core and the geodynamo.
Geophys. \& Astrophys. Fluid Dyn., 79:1-4, pp. 1-97. \\
\url{https://doi.org/10.1080/03091929508228992}


\bibitem{Brushlinsky1975}
Brushlinsky, K.V. 1975. 
Numerical simulation of two-dimensional plasma flow in channels. 
Computer Methods in Applied Mechanics and Engineering Vol. 6, pp. 293--308.
North-Holland Publishing Company.





\bibitem{Cotter-etal-2019}
Cotter, C.J.,  Crisan, D., Holm, D.D.,  Pan, W. and Shevchenko, I., 2019.
Numerically Modelling Stochastic Lie Transport in Fluid Dynamics,
SIAM Multiscale Model. Simul., 17(1), 192--232.\\
\url{https://doi.org/10.1137/18M1167929}

\bibitem{Cotter-etal-2020}
Cotter, C.J.,  Crisan, D., Holm, D.D.,  Pan, W. and Shevchenko, I., 2020.
Data Assimilation for a Quasi-Geostrophic Model with Circulation-Preserving Stochastic Transport Noise. 
\\J Stat Phys 179, 1186-1221. 
\url{https://doi.org/10.1007/s10955-020-02524-0}

\bibitem{Crisan-etal-2023a}
Crisan, D., Holm, D.D., Lang, O., Mensah, P.R. and Pan, W., 2023. 
Theoretical analysis and numerical approximation for the stochastic thermal quasi-geostrophic model. 
Stochastics and Dynamics, 23(05), p.2350039.
\url{https://doi.org/10.1142/S0219493723500399}

\bibitem{CrisanHolmLeahyNilssen2022}
Crisan, D., Holm, D.D., Leahy, J.M. and Nilssen, T., 2022. 
Variational principles for fluid dynamics on rough paths.
Advances in Mathematics, 404, p.108409.
\url{https://doi.org/10.1016/j.aim.2022.108409}


\bibitem{Crisan-etal-2023b}
Crisan, D., Holm, D.D., Luesink, E., Mensah, P. R. \& Pan, W., 2023. 
J. Nonlinear Sci. 33, \#96. \url{https://doi.org/10.1007/s00332-023-09943-9}

\bibitem{CHLN2022a}
Crisan, D., Holm, D.D., Leahy, J.M. and Nilssen, T., 2022. 
Variational principles for fluid dynamics on rough paths. 
Advances in Mathematics, 404, p.108409. 
\\
\url{https://doi.org/10.1016/j.aim.2022.108409} (arXiv:2004.07829).

\bibitem{CHLN2022b}
Crisan, D., Holm, D.D., Leahy, J.M. and Nilssen, T., 2022. 
Solution properties of the incompressible Euler system with rough path advection
Journal of Functional Analysis, 283 (9) 109632\\
\url{https://doi.org/10.1016/j.jfa.2022.109632} (arXiv:2104.14933).

\bibitem{DHP2023}
Diamantakis, T., Holm, D.D. and Pavliotis, G.A., 2023. 
Variational principles on geometric rough paths and the L\'evy area correction. 
SIAM Journal on Applied Dynamical Systems, 22(2), pp.1182-1218.
\url{https://doi.org/10.1137/22M1522164}	




\bibitem{D+H2020}
Drivas, T.D. and Holm, D.D., 2019. 
Circulation and Energy Theorem Preserving Stochastic Fluids.
Proc. Roy. Soc. Edinburgh A: Mathematics, 150 (6) 2776-2814.
\url{https://doi.org/10.1017/prm.2019.43}

\bibitem{DHL2020}
Drivas, T.D., Holm, D.D. and Leahy, J.M., 2020. 
Lagrangian averaged stochastic advection by Lie transport for fluids. 
Journal of Statistical Physics, 179(5-6), pp.1304-1342.
\url{https://doi.org/10.1007/s10955-020-02493-4}

\bibitem{Eyink2009}
Eyink, G., 2009. Stochastic line motion and stochastic flux conservation for nonideal hydromagnetic models,
J. Math. Phys. 50, 083102. \url{https://doi.org/10.1063/1.3193681}

\bibitem{Eyink-etal2013}
Eyink, G., Vishniac, E., Lalescu, C., Aluie, H., Kanov, K., Bürger, K., Burns, R., Meneveau, C. and Szalay, A., 2013. Flux-freezing breakdown in high-conductivity magnetohydrodynamic turbulence. Nature, 497(7450), pp.466-469. \url{https://doi.org/10.1038/nature12128}



\bibitem{FGB_H2014}
Gay-Balmaz, F. and Holm, D.D., 2014. A
geometric theory of selective decay with applications in MHD. 
Nonlinearity, 27(8), p.1747. \url{https://doi.org/10.1088/0951-7715/27/8/1747}





\bibitem{GV2021}
Gilbert, A.D. and Vanneste, J., 2021.
A geometric look at MHD and the Braginsky dynamo.
Geophys. \& Astrophys. Fluid Dyn., 115 (4) pp. 436--471. 
\url{https://doi.org/10.1080/03091929.2020.1839896}



\bibitem{Goedbloed2019}
Goedbloed, J.P., Goedbloed, H., Keppens, R. and Poedts, S., 2019. 
Magnetohydrodynamics: of laboratory and astrophysical plasmas. Cambridge University Press.

\bibitem{Gombosi_SW_2021}
Gombosi, T.I., Chen, Y., Glocer, A., Huang, Z., Jia, X., Liemohn, M.W., Manchester, W.B., Pulkkinen, T., Sachdeva, N., Al Shidi, Q. and Sokolov, I.V., 2021. What sustained multi-disciplinary research can achieve: The space weather modeling framework. Journal of Space Weather and Space Climate, 11, p.42.
	\url{https://doi.org/10.1051/swsc/2021020}



\bibitem{HanHuLai2024}
Han, B., Hu, K. and Lai, N.A., 2024. 
On the global well-posedness for the compressible Hall-MHD system. 
Journal of Mathematical Physics, 65(1). \url{https://doi.org/10.1063/5.0175649}

\bibitem{HMRW1985} 
Holm, D.D., Marsden, J.E., Ratiu, T.S. and  Weinstein, A. 1985. 
Nonlinear Stability of Fluid and Plasma Equilibria. 
\textit{ Physics Reports} \textbf{123} 1--116.\\
\url{https://doi.org/10.1016/0370-1573(85)90028-6}

\bibitem{Holm1985}
Holm, D.D., 1985. Hamiltonian Structure for 
Alfv\'en Wave Turbulence Equations.
\textit{ Phys. Lett. A} \textbf{108} (1985) 445--447.\\
\url{https://doi.org/10.1016/0375-9601(85)90035-0}

\bibitem{HHM1985} 
Hazeltine, R.D., Holm, D.D.  and Morrison, P.J., 1985.
Electromagnetic Solitary Waves in Magnetized Plasmas.
\textit{ J. Plasma Phys.} \textbf{34} 103--114.\\
\url{https://doi.org/10.1017/S0022377800002713}

\bibitem{Holm1987}
Holm, D.D., 1987.
Hall Magnetohydrodynamics: Conservation Laws and Lyapunov Stability,
\textit{ Phys.  Fluids} \textbf{30} (1987) 1310--1322.
\url{https://doi.org/10.1063/1.866246}

\bibitem{Holm2015}  
Holm, D.D., 2015.  Variational Principles for Stochastic Fluid Dynamics,  
\textit{Proc Roy Soc A}, 471: 20140963.   arXiv:1410.8311
\url{http://dx.doi.org/10.1098/rspa.2014.0963}

\bibitem{Holm2019}
Holm, D.D., 2019. Stochastic closures for wave-current interaction dynamics. Journal of
Nonlinear Science, 29 , 2987-3031. \url{https://doi.org/10.1007/s00332-019-09565-0}

\bibitem{Holm2024}
Holm, D.D., 2024. 
Transport Noise in Fluid Dynamics.
\textit{Institute of Mathematical Statistics Bulletin}
IMS Bulletin 53 (2), pp14--15, March 2024 issue. \\
\url{https://imstat.org/2024/02/15/youngstats-transport-noise-in-fluid-dynamics/}

\bibitem{HH2021}
Holm, D. D. \& Hu, R., 2021. Stochastic effects of waves on currents in the ocean mixed layer. Journal of Mathematical Physics 62 (7), 073102. \url{https://doi.org/10.1063/5.0045010}

\bibitem{HHS2022a}
Holm, D. D., Hu, R., \& Street, O. D., 2023. Coupling of Waves to Sea Surface Currents Via Horizontal Density Gradients. In: Chapron, B. et al. (eds) Stochastic Transport in Upper Ocean Dynamics. STUOD 2021. Mathematics of Planet Earth, vol 10. Springer, Cham. \url{https://doi.org/10.1007/978-3-031-18988-3_8}

\bibitem{HHS2022b}
Holm, D. D., Hu, R., \& Street, O. D., 2023. Lagrangian reduction and wave mean flow interaction. Physica D: Nonlinear Phenomena, 133847. \url{https://doi.org/10.1016/j.physd.2023.133847}


\bibitem{HolmKim1991}
Holm, D.D. and Kimura, Y., 1991. 
Zero-helicity Lagrangian Kinematics in Three-Dimensional Advection, 
Phys. Fluids A 3 (1991) 1033?1038. \url{https://doi.org/10.1063/1.858083}

\bibitem{HK1987} 
Holm, D.D. and Kupershmidt,  B.A., 1987.
Superfluid Plasmas: Multivelocity Nonlinear Hydrodynamics of Superfluid Solutions with Charged Condensates Coupled Electromagnetically, \\
\textit{ Phys. Rev. A} \textbf{36} (1987) 3947--3956.
\url{https://doi.org/10.1103/PhysRevA.36.3947}





\bibitem{HMR1998}
Holm, D.D., Marsden, J.E. and Ratiu, T.S., 1998. 
The Euler-Poincar\'e equations and semidirect products with applications to continuum theories. 
Advances in Mathematics, 137(1), pp.1-81. \\
\url{https://doi.org/10.1006/aima.1998.1721}




\bibitem{Huba2023}
Huba, J.D., 2023. Hall Magnetohydrodynamics. In: B\"uchner, J. (eds) Space and Astrophysical Plasma Simulation. Springer, Cham. \url{https://doi.org/10.1007/978-3-031-11870-8_2}

\bibitem{RHK1991}
Kraichnan, R.H., 1991.
Stochastic modeling of isotropic turbulence,
 In \textit{New Perspectives in Turbulence} - Springer.

\bibitem{SC2021}
Street, O. D., \& Crisan D., 2021. Semi-martingale driven variational principles. Proc. R. Soc. A.477 20200957 \url{http://doi.org/10.1098/rspa.2020.0957}

\bibitem{ST2023}
Street, O. D., \& Takao, S., 2023. Semimartingale driven mechanics and reduction by symmetry for stochastic and dissipative dynamical systems. Preprint. Under peer review. \url{https://arxiv.org/abs/2312.09769}

\bibitem{Troutman2012}
Troutman, J.L., 2012. Variational calculus and optimal control: optimization with elementary convexity. Springer Science \& Business Media.
 
 


\end{thebibliography}
\end{document}